\documentclass{article}
\usepackage[latin9]{inputenc}
\usepackage{mathtools}
\usepackage{amsmath}
\usepackage{amsthm}
\usepackage{amssymb}

\makeatletter
\theoremstyle{plain}
\newtheorem{thm}{\protect\theoremname}
  \theoremstyle{definition}
  \newtheorem{defn}[thm]{\protect\definitionname}
  \theoremstyle{remark}
  \newtheorem{rem}[thm]{\protect\remarkname}
  \theoremstyle{plain}
  \newtheorem{prop}[thm]{\protect\propositionname}
  \theoremstyle{plain}
  \newtheorem{lem}[thm]{\protect\lemmaname}
  \theoremstyle{plain}
  \newtheorem{cor}[thm]{\protect\corollaryname}

\usepackage{amsfonts}
\usepackage{comment}
\usepackage{graphicx}

\providecommand{\U}[1]{\protect\rule{.1in}{.1in}}
\newcommand{\Ad}{\operatorname{Ad}}
\newcommand{\ad}{\operatorname{ad}}

\newtheorem{theorem}{Theorem}\newtheorem{definition}[theorem]{Definition}\newtheorem{proposition}[theorem]{Proposition}\textwidth=410pt \oddsidemargin=20pt \voffset=-80pt
  \providecommand{\corollaryname}{Corollary}
  \providecommand{\definitionname}{Definition}
  \providecommand{\lemmaname}{Lemma}
  \providecommand{\propositionname}{Proposition}
  \providecommand{\remarkname}{Remark}
\providecommand{\theoremname}{Theorem}

\makeatother

  \providecommand{\corollaryname}{Corollary}
  \providecommand{\definitionname}{Definition}
  \providecommand{\lemmaname}{Lemma}
  \providecommand{\propositionname}{Proposition}
  \providecommand{\remarkname}{Remark}
\providecommand{\theoremname}{Theorem}

\begin{document}

\title{Variational reduction of Hamiltonian systems \\
with general constraints}

\author{Sergio D. Grillo\\
 Instituto Balseiro, U.N. de Cuyo-C.N.E.A.\\
 San Carlos de Bariloche, R\'io Negro, Rep\'ublica Argentina \and Leandro
M. Salomone\\
Centro de Matem\'atica de La Plata, Facultad de Ciencias Exactas, U.N.L.P.\\
 La Plata, Buenos Aires, Rep\'ublica Argentina \and Marcela Zuccalli\\
 Centro de Matem\'atica de La Plata, Facultad de Ciencias Exactas, U.N.L.P.\\
 La Plata, Buenos Aires, Rep\'ublica Argentina}
\maketitle
\begin{abstract}
In the Hamiltonian formalism, and in the presence of a symmetry Lie
group, a variational reduction procedure has already been developed
for Hamiltonian systems without constraints. In this paper we present
a procedure of the same kind, but for the entire class of the higher
order constrained systems (HOCS), described in the Hamiltonian formalism.
Last systems include the standard and generalized nonholonomic Hamiltonian
systems as particular cases. When restricted to Hamiltonian systems
without constraints, our procedure gives rise exactly to the so-called
Hamilton-Poincar\'e equations, as expected. In order to illustrate the
procedure, we study in detail the case in which both the configuration
space of the system and the involved symmetry define a trivial principal
bundle. 
\end{abstract}

\section{Introduction}

If a dynamical system defined on a manifold $M$ is invariant (in
some sense) under the action of a Lie group $G$, occasionally such
an invariance can be used to reduce the number, or at least the order,
of the differential equations that one must solve to find its trajectories.
More precisely, one can obtain, by using a certain procedure, a new
dynamical system on the quotient manifold $M/G$ such that: \textbf{1}.
the number of its equations of motion, the \textit{reduced equations},
is smaller than the number of the original equations of motion or,
at least, the order of some of the former is less than the order
of the latter; \textbf{2}. there exists another set of equations:
the \textit{reconstruction equations}, whose form does not depend
on the system under consideration (but only on $M$, $G$ and the
action involved) and which together with the reduced equations are
equivalent to the original ones. Thus, the symmetry, through a reduction
procedure, helps us to integrate (at least partially) the equations
of motion of the dynamical system originally given.

In some cases, in addition, a principal connection $A:TM\rightarrow\mathfrak{g}$,
where $TM$ is the tangent bundle of $M$ and $\mathfrak{g}$ is the
Lie algebra of $G$, can be constructed in order to simplify the description
of the reduced and reconstruction equations.

Since the reconstruction equations do not depend (in essence) on the
system, they are usually considered as a group theoretical problem.
Thus, philosophically, the problem of finding the trajectories of
the original system is considered solved when the solutions of the
reduced equations are found. This is why, from now on, we shall concentrate
mainly on the reduced equations only.

Reduction techniques have been developed by numerous authors in many
different frameworks. In particular, Cendra, Marsden and Ratiu elaborated
a reduction process for Lagrangian systems in Ref. \cite{cmr0} and
for (standard) nonholonomic systems in Ref. \cite{CMR}. For the case
of generalized nonholonomic systems (GNHS) (see \cite{solo,cg,her,Marle}),
a similar process was developed in \cite{cfg}, and an extension to
higher order constrained systems (HOCS) (see \cite{cg-hocs,her})
was presented in \cite{gz}. All of these reduction procedures were
elaborated in the Lagrangian formalism and in terms of variational-like
principles. The latter means that the original and the reduced equations
are described in terms of (original and reduced) variations and variational
conditions. Such variational conditions can in turn be translated
into a set of ordinary differential equations (ODE), in the same way
as the Hamilton principle is related to the Euler-Lagrange equations.
In terms of such ODEs, the number of reduced equations is equal to
the number of the original equations of motion (which are second order
ODEs), but some of the reduced equations are first order ODEs.

The Hamiltonian counterpart of above mentioned procedures is already
known for unconstrained Hamiltonian systems \cite{cmpr}. Nevertheless,
as far as we know, nothing have been done, in terms of variational-like
principles, for constrained Hamiltonian systems. In this paper, we
want to fill in this gap. In essence, to do that, we shall translate
to the Hamiltonian formalism the results obtained in Ref. \cite{gz},
extending in this way the procedure presented in Ref. \cite{cmpr}
to the class of all the HOCSs.

It is worth mentioning that, for our procedure (as happens in the
case of unconstrained systems \cite{cmpr}), the original and reduced
variational conditions are equivalent, each one of them, to a system
of first order ODEs. The result of reducing in the Hamiltonian formalism
is that the number of reduced ODEs is strictly less than that of
the original ones. Thus, by making a variational reduction in this
formalism, we effectively reduce the number of equations that we must
solve in order to find (modulo the reconstruction equations) the trajectories
of the system.

\bigskip{}

Along all of the paper we shall focus exclusively on Hamiltonian systems
defined on a cotangent bundle $T^{*}Q$, and symmetries given by a
Lie group $G$ acting on the base manifold $Q$. The actions of $G$
on $TQ$ and $T^{*}Q$ will be given by the corresponding canonical
lift.

\bigskip{}

The organization of the paper is as follows. In Section 2 we formulate
a variational reduction of a Hamiltonian GNHS (see Ref. \cite{cg})
with symmetry. Following Ref. \cite{gz}, we use two different principal
connections related to $Q$ and $G$, one to describe the reduced
degrees of freedom, i.e. the manifold $T^{*}Q/G$, and the other to
decompose the (original and reduced) variations into horizontal and
vertical parts. This will give rise to what we call the \textit{Hamilton-d'Alembert-Poincar\'e
horizontal and vertical equations}. The latter, in absence of constraints,
correspond exactly to the \textit{Hamilton-Poincar\'e equations}, described
in Ref. \cite{cmpr}. In Section 3 we consider the Hamiltonian HOCSs
(see Ref. \cite{g2}) with symmetry and develop a reduction process
for them. In order to do that, we define a connection-like object
called \textit{cotangent $l$-connection} (the dual notion of tangent
$l$-connection introduced in Ref. \cite{gz}). Finally, in Section
4 we study the case in which $Q$ and $G$ define a trivial principal
bundle, and in Section 5 we present an illustrative example of that
case: a ball rolling over another ball.

\bigskip{}

We assume that the reader is familiar with the basic concepts of Differential
Geometry (see \cite{boot,kn,mrgm}) as well as the ideas of Lagrangian
and Hamiltonian systems with symmetry in the context of Geometric
Mechanics (see \cite{am,mr}).

\bigskip{}

\paragraph*{A word about notation.}

Throughout all of the paper, $Q$ will denote a differentiable finite-dimensional
manifold. The tangent and cotangent bundles of $Q$ will be denoted
as $\tau_{Q}:TQ\rightarrow Q$ and $\pi_{Q}:T^{*}Q\rightarrow Q$,
respectively. If $M$ is another differentiable manifold and $f:Q\rightarrow M$
is a differentiable function, we denote by $f_{*}:TQ\rightarrow TM$
and $f^{*}:T^{*}M\rightarrow T^{*}Q$ the tangent map and its transpose,
respectively. Given two fibrations $E_{1}\rightarrow Q$ and $E_{2}\rightarrow Q$,
we denote by $E_{1}\times_{Q}E_{2}$ the fibered product of $E_{1}$
and $E_{2}$ over $Q$. If $E_{1}$ and $E_{2}$ are vector bundles,
then $E_{1}\times_{Q}E_{2}$ is also a vector bundle with respect
to the natural component-wise linear structure. When
such a structure is taking into account, the product bundle is called
the Whitney sum of $E_{1}$ and $E_{2}$, and will be denoted $E_{1}\oplus E_{2}$
or $E_{1}\oplus_{Q}E_{2}$.

\section{Reduction of Hamiltonian GNHS}

\label{redgnhs}

In this section we recall the definition of a generalized nonholonomic
system (GNHS) in the Hamiltonian formalism, following References \cite{cg,Marle}.
Then, given a Lie group $G$, we define the idea of $G$-invariance
for such systems. Finally, for the $G$-invariant ones, we develop
two different reduction procedures such that (both of them): 
\begin{enumerate}
\item extend the one presented in \cite{cmpr} (which is only valid for
unconstrained Hamiltonian systems); 
\item represent the Hamiltonian counterpart of the procedures elaborated
in Refs. \cite{cfg,gz}. 
\end{enumerate}

\subsection{Hamiltonian GNHS with symmetry}

\label{basic}

Motivated by certain mechanical systems such as rubber wheels and
servomechanisms, where d'Alembert principle is typically violated,
it was defined and studied in Refs. \cite{solo,cg,her,Marle} a class
of constrained Lagrangian systems, called generalized nonholonomic
system (GNHS), that include the above mentioned mechanical systems
and encode, in our opinion, their main features. In the Hamiltonian
framework, they can be defined as follows (see \cite{cg} and \cite{Marle}).
Consider a triple $\left(H,D,\mathcal{V}\right)$ with 
\[
H:T^{\ast}Q\rightarrow\mathbb{R},\quad D\subset T^{\ast}Q\quad\text{and}\quad\mathcal{V}\subset TT^{\ast}Q,
\]
where $H$ is a differentiable function, $D$ is a submanifold of
$T^{\ast}Q$ and $\mathcal{V}$ is a vector subbundle of the tangent
bundle $TT^{\ast}Q$. 
\begin{defn}
We shall say that a triple $\left(H,D,\mathcal{V}\right)$ as above
is a \textbf{Hamiltonian }\textbf{\small{}{}GNHS}{\small{} on the
configuration space $Q$, with }\textbf{\small{}Hamiltonian function}{\small{}
$H$, }\textbf{\small{}kinematic constraints}{\small{} $D$ and }\textbf{\small{}variational
constraints}{\small{} $\mathcal{V}$. And we shall say that a curve
$\Gamma:\left[t_{1},t_{2}\right]\rightarrow T^{\ast}Q$ is a }\textbf{\small{}trajectory}{\small{}
of $\left(H,D,\mathcal{V}\right)$ if $\Gamma\left(t\right)\in D$
and for all infinitesimal variations with fixed end points}\footnote{Recall that, given a manifold $M$ and a curve $\sigma:\left[t_{1},t_{2}\right]\rightarrow M$,
an \textbf{infinitesimal variation} of $\sigma$ is a curve $\delta\sigma\colon\left[t_{1},t_{2}\right]\rightarrow TM$
satisfying $\delta\sigma\left(t\right)\in T_{\sigma\left(t\right)}M$,
$\forall\,t\in\left[t_{1},t_{2}\right]$. We say that $\delta\sigma$
has \textbf{fixed end points} if $\delta\sigma\left(t_{1}\right)$
and $\delta\sigma\left(t_{2}\right)$ belong to the null subbundle
of $TM$.}{\small{} $\delta\Gamma:\left[t_{1},t_{2}\right]\rightarrow TT^{\ast}Q$,
such that $\delta\Gamma\left(t\right)\in\mathcal{V}$, we have 
\begin{equation}
\int_{t_{1}}^{t_{2}}\left(\omega\left(\Gamma^{\prime}\left(t\right),\delta\Gamma\left(t\right)\right)-\left\langle dH\left(\Gamma\left(t\right)\right),\delta\Gamma\left(t\right)\right\rangle \right)\ dt=0.\label{eve}
\end{equation}
}{\small \par}
\end{defn}
By $\Gamma^{\prime}:\left(t_{1},t_{2}\right)\rightarrow TT^{\ast}Q$
we are denoting the \emph{velocity} of $\Gamma$, defined as 
\[
\Gamma^{\prime}\left(t\right):=\frac{d}{dt}\Gamma\left(t\right)\in T_{\Gamma\left(t\right)}T^{\ast}Q.
\]
As usual, $\omega:TT^{\ast}Q\times_{T^{\ast}Q}TT^{\ast}Q\rightarrow\mathbb{R}$
denotes the canonical symplectic $2$-form of $T^{\ast}Q$. Then,
$\omega=-d\theta$, being $\theta:TT^{\ast}Q\rightarrow\mathbb{R}$
the canonical $1$-form of $T^{\ast}Q$, given by 
\[
\theta\left(V\right):=\left\langle \tau_{T^{\ast}Q}\left(V\right),\pi_{Q\ast}\left(V\right)\right\rangle ,\ \ \ \forall\,V\in TT^{\ast}Q.
\]

Let us note that Eq. \eqref{eve} is an extremal condition for the
action functional 
\begin{equation}
S\left(\Gamma\right):=\int_{t_{1}}^{t_{2}}\left[\theta\left(\Gamma^{\prime}\left(t\right)\right)-H\left(\Gamma\left(t\right)\right)\right]\ dt=\int_{t_{1}}^{t_{2}}\left[\left\langle \Gamma\left(t\right),\pi_{Q\ast}\left(\Gamma^{\prime}\left(t\right)\right)\right\rangle -H\left(\Gamma\left(t\right)\right)\right]\ dt,\label{S}
\end{equation}
for variations $\delta\Gamma$ inside $\mathcal{V}$. In fact, any
variation $\delta\Gamma$ can be defined by a map 
\[
\left[t_{1},t_{2}\right]\times\left(-\epsilon,\epsilon\right)\rightarrow T^{\ast}Q:\left(t,s\right)\mapsto\Gamma_{s}\left(t\right),
\]
such that 
\[
\Gamma_{0}\left(t\right)=\Gamma\left(t\right)\ \ \ \text{and\ \ \ }\left.\frac{\partial}{\partial s}\right\vert _{0}\Gamma_{s}\left(t_{1,2}\right)=0,
\]
through the formula 
\[
\delta\Gamma\left(t\right):=\left.\frac{\partial}{\partial s}\right\vert _{0}\Gamma_{s}\left(t\right).
\]
So, using the equality $d\theta=-\omega$ (and the fixed end point
conditions $\delta\Gamma\left(t_{1,2}\right)=0$) we have 
\begin{align*}
\left.\frac{\partial}{\partial s}\right\vert _{0}S\left(\Gamma_{s}\right) & =\int_{t_{1}}^{t_{2}}\left[\left.\frac{\partial}{\partial s}\right\vert _{0}\theta\left(\Gamma_{s}^{\prime}\left(t\right)\right)-\left.\frac{\partial}{\partial s}\right\vert _{0}H\left(\Gamma_{s}\left(t\right)\right)\right]\ dt\\
\\
 & =\int_{t_{1}}^{t_{2}}\left(\omega\left(\Gamma^{\prime}\left(t\right),\delta\Gamma\left(t\right)\right)-\left\langle dH\left(\Gamma\left(t\right)\right),\delta\Gamma\left(t\right)\right\rangle \right)\ dt.
\end{align*}

\bigskip{}

We shall make the following assumptions for the variational constraints
$\mathcal{V}$: 
\begin{description}
\item [{A1}] The subbundle $\mathcal{V}^{\perp}$, the symplectic orthogonal
of $\mathcal{V}$, is a vertical subbundle; that is, $\mathcal{V}^{\perp}\subset\ker\left(\pi_{Q\ast}\right)$. 
\item [{A2}] The subset 
\[
C_{V}:=\pi_{Q\ast}\left(\mathcal{V}\right)\subset TQ
\]
defines a vector subbundle of $TQ$. 
\end{description}
It can be shown that \textbf{A1} implies the inclusions

\[
\mathcal{V}^{\perp}\subset\ker\left(\pi_{Q\ast}\right)\subset\mathcal{V},
\]
and consequently 
\begin{equation}
v\in\mathcal{V}\ \ \ \mbox{if and only if}\ \ \pi{}_{Q\ast}\left(v\right)\in\pi_{Q\ast}\left(\mathcal{V}\right).\label{ioi}
\end{equation}

\begin{rem}
Assumptions \textbf{A1} is related to the physical meaning of $\mathcal{V}^{\perp}$:
the \textbf{space of the constraint forces} (for details see Refs.
\cite{cg} and \cite{g}). It says that the forces are given by vertical
vectors. On the other hand, assumption \textbf{A2 }says that the space
of constraint forces does not depend on velocities, but only on positions. 
\end{rem}
Define $\gamma\left(t\right):=\pi_{Q}\left(\Gamma\left(t\right)\right)$.
It is clear that $\delta\gamma\left(t\right):=\pi_{Q\ast}\left(\delta\Gamma\left(t\right)\right)$
is an infinitesimal variation of $\gamma$. Using the Eq. \eqref{ioi}
and assumption \textbf{A2}, it easily follows that 
\begin{equation}
\delta\Gamma\left(t\right)\in\left.\mathcal{V}\right\vert _{\Gamma\left(t\right)}\ \Leftrightarrow\ \delta\gamma\left(t\right)\in\left.C_{V}\right\vert _{\gamma\left(t\right)}.\label{eqa}
\end{equation}
On the other hand, in terms of $\gamma$, the action defined in Eq.
\eqref{S} can be written 
\begin{equation}
S\left(\Gamma\right)=\int_{t_{1}}^{t_{2}}\left[\left\langle \Gamma\left(t\right),\gamma^{\prime}\left(t\right)\right\rangle -H\left(\Gamma\left(t\right)\right)\right]\ dt.\label{S2}
\end{equation}
Thus, Eq. \eqref{eve} is an extremal condition for \eqref{S2} for
variations $\delta\gamma$ inside $C_{V}$. 
\begin{rem}
\label{hdcv} By assuming conditions \textbf{A1} and \textbf{A2},
a Hamiltonian GNHS $\left(H,D,\mathcal{V}\right)$ can be completely
described by the data $\left(H,D,C_{V}\right)$, and we shall do it
from now on. The cases in which \textbf{A2} is not satisfied will be
studied in Section 3, in the context of higher order constrained systems. 
\end{rem}
Now suppose that we have a Lie group $G$ acting on $Q$ through a
map $\rho:G\times Q\rightarrow Q$. (We choose here to work with left
actions, but for right actions we would have similar definitions and
results). Consider the canonical lifted actions of $\rho$ to $TQ$
and $T^{\ast}Q$, given by 
\[
G\times TQ\rightarrow TQ:\left(g,v\right)\mapsto\left(\rho_{g}\right)_{\ast}\left(v\right)
\]
and 
\begin{equation}
G\times T^{\ast}Q\rightarrow T^{\ast}Q:\left(g,\sigma\right)\mapsto\hat{\rho}_{g}\left(\sigma\right):=\left(\rho_{g^{-1}}\right)^{\ast}\left(\sigma\right),\label{rhotilde}
\end{equation}
respectively, where the diffeomorphism $\rho_{g}:Q\rightarrow Q$
is given by $\rho_{g}(q)=\rho\left(g,q\right)$. Notice that 
\begin{equation}
\rho_{g}\circ\pi_{Q}=\pi_{Q}\circ\hat{\rho}_{g},\;\;\;\forall g\in G.\label{lad}
\end{equation}
We shall assume that $\mathcal{X}\coloneqq Q/G$, $TQ/G$ and $T^{\ast}Q/G$
are manifolds and that the canonical projections $\pi:Q\rightarrow\mathcal{X}$,
$p:TQ\rightarrow TQ/G$ and $\hat{p}:T^{\ast}Q\rightarrow T^{\ast}Q/G$
are submersions.

\begin{definition}We shall say that a Hamiltonian GNHS $\left(H,D,C_{V}\right)$
is $G$-\textbf{invariant} if for all $g\in G$: 
\begin{description}
\item [{a.}] $H\circ\hat{\rho}_{g}=H$, 
\item [{b.}] $\hat{\rho}_{g}\left(D\right)=D$ and 
\item [{c.}] $\left(\rho_{g}\right)_{\ast}\left(C_{V}\right)=C_{V}$. 
\end{description}
In this case, we shall also say that the Lie group $G$ is a symmetry
of the triple $\left(H,D,C_{V}\right)$. \end{definition}

By using the canonical projections $p$ and $\hat{p}$ described above,
we can define the \textbf{reduced Hamiltonian} $h:T^{\ast}Q/G\rightarrow\mathbb{R}$,
given by 
\begin{equation}
h\circ\hat{p}=H,\label{Hh}
\end{equation}
and the \textbf{reduced kinematic and variational constraints} 
\begin{equation}
\mathfrak{D}:=\hat{p}\left(D\right)=D/G\ \ \ \ \text{and}\ \ \ \ \mathfrak{C}_{V}:=p\left(C_{V}\right)=C_{V}/G,\label{cr}
\end{equation}
respectively.

\subsection{A reduction procedure with one connection }

The aim of this subsection is to write down the equations of motion
of $\left(H,D,C_{V}\right)$ in terms of the reduced data $h$, $\mathfrak{D}$
and $\mathfrak{C}_{V}$. In order to do that, we shall consider the
results presented in \cite{cfg} and \cite{cmpr}. In particular,
we shall use the so-called generalized nonholonomic connection, defined
from the variational constraints.

\subsubsection{The Atiyah isomorphism}

\label{gnhc} From now on, we shall assume that the action $\rho:G\times Q\rightarrow Q$
is free, what implies that $\pi:Q\rightarrow\mathcal{X}$ is a principal
fiber bundle. Recall that a principal connection for $\pi:Q\rightarrow\mathcal{X}$
is a map $A:TQ\rightarrow\mathfrak{g}$ such that 
\begin{equation}
A\left(\eta_{Q}\left(q\right)\right)=\eta\quad\text{and}\quad A\left(\rho_{g\ast}\left(v\right)\right)=\Ad_{g}A\left(v\right),\label{cpc}
\end{equation}
where $\mathfrak{g}$ is the Lie algebra of $G$, $\eta_{Q}\in\mathfrak{X}\left(Q\right)$
is the fundamental vector field related to $\eta\in\mathfrak{g}$,
and $\Ad_{g}:\mathfrak{g}\rightarrow\mathfrak{g}$ is the adjoint
action of $g\in G$ on the Lie algebra $\mathfrak{g}$. It is well-known
that $A$ gives rise to a fiber bundle isomorphism (see Ref. \cite{cmr0})
\[
\alpha_{A}:TQ/G\rightarrow T\mathcal{X}\oplus\widetilde{\mathfrak{g}},
\]
called \emph{Atiyah isomorphism,} given by 
\[
\alpha_{A}\left(\left[v\right]\right):=\left(\pi_{\ast}\left(v\right),\left[q,A\left(v\right)\right]\right),\quad\text{for all}\quad q\in Q\quad\text{and}\quad v\in T_{q}Q,
\]
where $\left[v\right]:=p\left(v\right)\in\left.TQ\right/G$. By $\widetilde{\mathfrak{g}}:=\left(Q\times\mathfrak{g}\right)/G$
we are denoting the adjoint bundle (with base $\mathcal{X}$). (The
action of $G$ on $Q\times\mathfrak{g}$ is given by the action $\rho$
on $Q$ and the adjoint action on $\mathfrak{g}$). The elements of
$\widetilde{\mathfrak{g}}$ are denoted as equivalence classes $\left[q,\eta\right]$,
with $q\in Q$ and $\eta\in\mathfrak{g}$.

For later convenience, note that defining 
\[
a:TQ\rightarrow\widetilde{\mathfrak{g}}\ :\ v\mapsto\left[q,A\left(v\right)\right],
\]
we have that 
\begin{equation}
\alpha_{A}\circ p\left(v\right)=\pi_{\ast}\left(v\right)\oplus a\left(v\right),\ \ \ \forall v\in TQ.\label{pma}
\end{equation}

Related to $\alpha_{A}$, we have the next results. 
\begin{itemize}
\item Since $\alpha_{A}$ is a vector bundle isomorphism, then for each
$q\in Q$ the spaces $T_{\pi\left(q\right)}\mathcal{X}\oplus\widetilde{\mathfrak{g}}_{\pi\left(q\right)}$
and $\left(\left.TQ\right/G\right)_{\pi\left(q\right)}$ have the
same dimension. Moreover, since $\rho$ is a free action, it can be
shown that the map 
\[
\alpha_{A}\circ p:TQ\rightarrow T\mathcal{X}\oplus\widetilde{\mathfrak{g}}
\]
defines a linear isomorphism between $T_{q}Q$ and $T_{\pi(q)}X\oplus\mathfrak{\widetilde{g}}_{\pi(q)}$
when restricted to each fiber $T_{q}Q$. 
\item Let $\mathbb{H}$ denote the horizontal subbundle related to $A$
and $\mathbb{V}\coloneqq\ker\pi_{*}$ the vertical subbundle.
For each $q\in Q$, the restrictions of $\alpha_{A}\circ p$ to $\mathbb{H}_{q}$
and $\mathbb{V}_{q}$ are injective and 
\begin{equation}
\alpha_{A}\circ p\left(\mathbb{H}_{q}\right)=\pi_{\ast}\left(\mathbb{H}_{q}\right)=T_{\pi\left(q\right)}\mathcal{X}\label{eh}
\end{equation}
and 
\begin{equation}
\alpha_{A}\circ p\left(\mathbb{V}_{q}\right)=a\left(\mathbb{V}_{q}\right)=\widetilde{\mathfrak{g}}_{\pi\left(q\right)}.\label{ev}
\end{equation}
\item By identifying the bundles $\left(TQ/G\right)^{*}$ and $\left.T^{\ast}Q\right/G$
in a canonical way, we can define the fiber bundle isomorphism $\hat{\alpha}_{A}:T^{\ast}Q/G\rightarrow T^{\ast}\mathcal{X}\oplus{\widetilde{\mathfrak{g}}}^{*}$
given by 
\[
\hat{\alpha}_{A}:=\left(\alpha_{A}^{-1}\right)^{\ast},
\]
where $(\widetilde{\mathfrak{g}})^{\ast}$ is identified with $\widetilde{(\mathfrak{g}^{\ast})}$
in a natural way. 
\item Again, for each $q\in Q$, the linear spaces $\left(\left.T^{*}Q\right/G\right)_{\pi\left(q\right)}$
and $T_{\pi\left(q\right)}^{*}\mathcal{X}\oplus\widetilde{\mathfrak{g}}_{\pi\left(q\right)}^{*}$
have the same dimension and, since $\rho$ is a free action, the map
\[
\hat{\alpha}_{A}\circ\hat{p}:T^{\ast}Q\rightarrow T^{\ast}\mathcal{X}\oplus\widetilde{\mathfrak{g}}^{\ast}
\]
defines a linear isomorphism with its image when restricted to each
$T_{q}^{*}Q$. 
\end{itemize}
\begin{rem}
\label{momentum} As for a standard (unconstrained) Hamiltonian system,
if the action of $G$ on $Q$ preserves the symplectic form (see Ref.
\cite{am} for more details), we have an application ${\bf J}:T^{*}Q\rightarrow\mathfrak{g}^{*}$,
called \textit{momentum map,}\footnote{The function $\mathbf{J}$ is a conserved quantity in the case of
unconstrained systems. In the case, for instance, of nonholonomic
systems, the momentum map is not conserved in general. Nevertheless,
it remains a relevant datum of the system and it is possible to know
its evolution along the trajectories \cite{bl96}.} defined (at least locally) by the formula 
\[
{\displaystyle \left<{\bf J}(\sigma_{q}),\eta\right>:=J(\eta)(\sigma_{q})=\left<\sigma_{q},\eta_{Q}(q)\right>},\;\;\;\forall\,\eta\in\mathfrak{g},
\]
where $J(\eta)$ is a smooth function on $T^{*}Q$ such that 
\[
i_{\eta_{Q}}\omega=dJ(\eta).
\]
By using this application, the isomorphism $\hat{\alpha}_{A}:T^{\ast}Q/G\rightarrow T^{\ast}\mathcal{X}\oplus{\widetilde{\mathfrak{g}}}^{*}$
can be written as 
\[
\hat{\alpha}_{A}([\sigma_{q}])=\left(hor_{q}^{*}\sigma_{q},[q,{\bf J}(\sigma_{q})]\right),
\]
where $hor_{q}^{*}:T_{q}^{*}Q\rightarrow T_{\pi(q)}^{*}\mathcal{X}$
is dual to the horizontal lift map $hor_{q}:T_{\pi(q)}\mathcal{X}\rightarrow T_{q}Q$
associated to the connection $A$. 
\end{rem}

\subsubsection{Generalized nonholonomic connection: \emph{reduced} horizontal and
vertical variations}

Given a $G$-invariant Hamiltonian GNHS $\left(H,D,C_{V}\right)$,
if $A$ is an arbitrary principal connection on $\pi$, the subset
\[
\alpha_{A}\circ p\left(C_{V}\right)=\alpha_{A}\left(\mathfrak{C}_{V}\right)\subset T\mathcal{X}\oplus\widetilde{\mathfrak{g}}
\]
defines a vector subbundle of $T\mathcal{X}\oplus\widetilde{\mathfrak{g}}$
whose elements can be called \emph{reduced variations}. Let us identify
$\mathfrak{C}_{V}$ and $\alpha_{A}\left(\mathfrak{C}_{V}\right)$,
i.e. let us write 
\[
\mathfrak{C}_{V}:=\alpha_{A}\circ p\left(C_{V}\right).
\]
As in Ref. \cite{cfg}, consider the \textit{generalized nonholonomic
connection} $A^{\bullet}:TQ\rightarrow\mathfrak{g}$ defined from
the variational constraints. If $\mathbb{H}^{\bullet}$ denotes the
horizontal subbundle related to $A^{\bullet}$, we can write 
\[
C_{V}=\left(C_{V}\cap\mathbb{H}^{\bullet}\right)\oplus\left(C_{V}\cap\mathbb{V}\right)
\]
and 
\[
\mathfrak{C}_{V}^{\bullet}:=\alpha_{A^{\bullet}}\circ p\left(C_{V}\right)
\]
By using \eqref{eh} and \eqref{ev}, we can prove that 
\[
\mathfrak{C}_{V}^{\bullet}=\mathfrak{C}_{V}^{\text{\textrm{hor}}}\oplus\mathfrak{C}_{V}^{\text{\textrm{ver}}},
\]
where 
\begin{equation}
\mathfrak{C}_{V}^{\text{\textrm{hor}}}:=\pi_{\ast}\left(C_{V}\right)=\mathfrak{C}_{V}^{\bullet}\cap T\mathcal{X}\label{ttt}
\end{equation}
and 
\begin{equation}
\mathfrak{C}_{V}^{\text{\textrm{ver}}}:=a^{\bullet}\left(C_{V}\right)=\mathfrak{C}_{V}^{\bullet}\cap\widetilde{\mathfrak{g}}.\label{sss}
\end{equation}
That is, using the connection $A^{\bullet}$, the reduced variations
decompose into horizontal and vertical parts which are mutually independent. 
\begin{rem}
\label{id} As we did with $\mathfrak{C}_{V}$, we shall see the reduced
kinematic constraint $\mathfrak{D}$ {[}see \eqref{cr}{]} as a subset
of $T^{\ast}\mathcal{X}\oplus\widetilde{\mathfrak{g}}^{\ast}$, i.e.
we shall make the identification 
\[
\mathfrak{D}=\hat{\alpha}_{A}\circ\hat{p}\left(D\right).
\]
Moreover, from now on, and if there is no risk of confusion, we shall
identify the fiber bundles $\left.TQ\right/G$ (resp. $\left.T^{\ast}Q\right/G$)
and $T\mathcal{X}\oplus\widetilde{\mathfrak{g}}$ (resp. $T^{\ast}\mathcal{X}\oplus\widetilde{\mathfrak{g}}^{\ast}$)
\emph{via} the map $\alpha_{A}$ (resp. $\hat{\alpha}_{A}$). If $A$
is an arbitrary principal connection, we will write $\mathfrak{C}_{V}$
and $\mathfrak{D}$. If, on the other hand, we use the generalized
nonholonomic connection $A^{\bullet}$, we shall write $\mathfrak{C}_{V}^{\bullet}$
and $\mathfrak{D}^{\bullet}$, respectively. 
\end{rem}

\subsubsection{Reduced variational principle}

\label{reducedextremal}

As we have seen in Section \ref{basic}, a curve $\Gamma:\left[t_{1},t_{2}\right]\rightarrow D\subset T^{\ast}Q$
is a trajectory of $\left(H,D,C_{V}\right)$ if and only if it is
an extremal of the action 
\[
S\left(\Gamma\right)=\int_{t_{1}}^{t_{2}}\left[\left\langle \Gamma\left(t\right),\gamma^{\prime}\left(t\right)\right\rangle -H\left(\Gamma\left(t\right)\right)\right]\ dt
\]
for all variations $\delta\Gamma$ such that $\delta\gamma$ lies
on $C_{V}$. We want to write this extremal condition in terms of
the reduced data $\left(h,\mathfrak{D},\mathfrak{C}_{V}\right)$.

\bigskip{}

Fix a principal connection $A$ (we are not assuming at this point
that $A=A^{\bullet}$). Using the identification between $T^{\ast}Q/G$
and $T^{\ast}\mathcal{X}\oplus\widetilde{\mathfrak{g}}^{\ast}$ given
by $\hat{\alpha}_{A}$, let us denote the composition 
\[
h\circ\hat{\alpha}_{A}^{-1}:T^{\ast}\mathcal{X}\oplus\widetilde{\mathfrak{g}}^{\ast}\rightarrow\mathbb{R}
\]
simply as $h$. Consider a curve $\Gamma:\left[t_{1},t_{2}\right]\rightarrow T^{\ast}Q$
and write $\pi_{Q}\left(\Gamma(t)\right)\eqqcolon\gamma(t)$. Following
the notation of \cite{cmpr} and \cite{cmr0}, let us define 
\begin{equation}
\ \ \ \ x\left(t\right):=\pi\left(\gamma\left(t\right)\right),\ \ \ \dot{x}\left(t\right)\oplus\bar{v}\left(t\right):=\pi_{\ast}\left(\gamma^{\prime}\left(t\right)\right)\oplus a\left(\gamma^{\prime}\left(t\right)\right)=\alpha_{A}\circ p\left(\gamma^{\prime}\left(t\right)\right)\label{def}
\end{equation}
and 
\begin{equation}
\varsigma\left(t\right):=y\left(t\right)\oplus\bar{\mu}\left(t\right):=\hat{\alpha}_{A}\circ\hat{p}\left(\Gamma\left(t\right)\right).\label{seda}
\end{equation}
Then, recalling \eqref{Hh}, it is easy to show that 
\[
\left\langle \Gamma\left(t\right),\gamma^{\prime}\left(t\right)\right\rangle -H\left(\Gamma\left(t\right)\right)=\left\langle y\left(t\right),\dot{x}\left(t\right)\right\rangle +\left\langle \bar{\mu}\left(t\right),\bar{v}\left(t\right)\right\rangle -h\left(\varsigma\left(t\right)\right)
\]
and consequently 
\[
S\left(\Gamma\right)=\int_{t_{1}}^{t_{2}}\left[\left\langle y\left(t\right),\dot{x}\left(t\right)\right\rangle +\left\langle \bar{\mu}\left(t\right),\bar{v}\left(t\right)\right\rangle -h\left(\varsigma\left(t\right)\right)\right]\ dt.
\]
Now, let $\nabla^{A}$ be the affine connection induced by $A$ in
$\widetilde{\mathfrak{g}}$ and $\widetilde{\mathfrak{g}}^{*}$ and
fix an affine connection $\nabla^{\mathcal{X}}$ on $\mathcal{X}$.
Also, denote by $B:TQ\times_{Q}TQ\rightarrow\mathfrak{g}$ the curvature
of $A$. Given a variation 
\[
\delta\Gamma\left(t\right)=\left.\frac{\partial}{\partial s}\right\vert _{0}\Gamma_{s}\left(t\right)
\]
with fixed end points, it can be shown that 
\begin{align*}
\left.\frac{\partial}{\partial s}\right\vert _{0}S\left(\Gamma_{s}\right) & =\int_{t_{1}}^{t_{2}}[\left\langle \delta y\left(t\right),\dot{x}\left(t\right)\right\rangle +\left\langle y\left(t\right),\delta\dot{x}\left(t\right)\right\rangle +\\
\\
 & +\left\langle \delta\bar{\mu}\left(t\right),\bar{v}\left(t\right)\right\rangle +\left\langle \bar{\mu}\left(t\right),\delta\bar{v}\left(t\right)\right\rangle -\delta h\left(\varsigma\left(t\right)\right)]\ dt
\end{align*}
for some curves $\delta y:\left[t_{1},t_{2}\right]\rightarrow T^{\ast}\mathcal{X}$
and $\delta\bar{\mu}:\left[t_{1},t_{2}\right]\rightarrow\widetilde{\mathfrak{g}}^{\ast}$,
and where 
\[
\delta x\left(t\right)=\pi_{\ast}\left(\delta\gamma\left(t\right)\right),\ \ \ \bar{\eta}\left(t\right)=a\left(\delta\gamma\left(t\right)\right),
\]
\[
\delta\dot{x}\left(t\right)=\frac{D\delta x}{Dt}\left(t\right),\ \ \ \delta\bar{v}\left(t\right)=\frac{D\bar{\eta}}{Dt}\left(t\right)+\left[\bar{v}\left(t\right),\bar{\eta}\left(t\right)\right]-\widetilde{B}\left(\dot{x}\left(t\right),\delta x\left(t\right)\right),
\]
and 
\begin{equation}
\widetilde{B}:T\mathcal{X}\times_{\mathcal{X}}T\mathcal{X}\rightarrow\widetilde{\mathfrak{g}}\ :\ \left(\pi_{\ast}\left(u_{q}\right),\pi_{\ast}\left(v_{q}\right)\right)\mapsto\left[q,B\left(u_{q},v_{q}\right)\right]\label{bt}
\end{equation}
is the \emph{reduced curvature }of $A$. Also, we can write 
\[
\delta h\left(\varsigma\left(t\right)\right)=\left\langle \frac{\partial h}{\partial x}\left(\varsigma\left(t\right)\right),\delta x\left(t\right)\right\rangle +\left\langle \delta y\left(t\right),\frac{\partial h}{\partial y}\left(\varsigma\left(t\right)\right)\right\rangle +\left\langle \delta\bar{\mu}\left(t\right),\frac{\partial h}{\partial\bar{\mu}}\left(\varsigma\left(t\right)\right)\right\rangle 
\]
where 
\[
\frac{\partial h}{\partial y}:T^{\ast}\mathcal{X}\oplus\widetilde{\mathfrak{g}}^{\ast}\rightarrow T\mathcal{X}\ \ \ \text{and}\ \ \frac{\partial h}{\partial\bar{\mu}}:T^{\ast}\mathcal{X}\oplus\widetilde{\mathfrak{g}}^{\ast}\rightarrow\widetilde{\mathfrak{g}}
\]
are the first and second components of the fiber derivative 
\[
\mathbb{F}h:T^{\ast}\mathcal{X}\oplus\widetilde{\mathfrak{g}}^{\ast}\rightarrow T\mathcal{X}\oplus\widetilde{\mathfrak{g}}
\]
of $h$ and 
\[
\frac{\partial h}{\partial x}:T^{\ast}\mathcal{X}\oplus\widetilde{\mathfrak{g}}^{\ast}\rightarrow T^{\ast}\mathcal{X}
\]
its base derivative with respect to an affine connection $\nabla^{\mathcal{X}}\oplus\nabla^{A}$.
See \cite{cmpr} and \cite{cmr0} for more details. Accordingly, integrating
by parts and using the fixed end points condition for $\delta\Gamma$,
\[
\left.\frac{\partial}{\partial s}\right\vert _{0}S\left(\Gamma_{s}\right)=0\ \ \ \mbox{if and only if}\ \ \ 
\]
\begin{align}
 & \left\langle -\frac{D}{Dt}y\left(t\right)-\frac{\partial h}{\partial x}\left(\varsigma\left(t\right)\right)-\left\langle \bar{\mu}\left(t\right),\widetilde{B}\left(\dot{x}\left(t\right),\cdot\right)\right\rangle ,\delta x\left(t\right)\right\rangle +\left\langle \delta y\left(t\right),\dot{x}\left(t\right)-\frac{\partial h}{\partial y}\left(\varsigma\left(t\right)\right)\right\rangle \nonumber \\
\nonumber \\
 & +\left\langle -\frac{D}{Dt}\bar{\mu}\left(t\right)+\ad_{\bar{v}\left(t\right)}^{\ast}\bar{\mu}\left(t\right),\bar{\eta}\left(t\right)\right\rangle +\left\langle \delta\bar{\mu}\left(t\right),\bar{v}\left(t\right)-\frac{\partial h}{\partial\bar{\mu}}\left(\varsigma\left(t\right)\right)\right\rangle =0,\label{rv}
\end{align}
where $\ad_{[q,v]}^{*}$ is the transpose of the map $\ad_{[q,v]}:\widetilde{\mathfrak{g}}_{\pi(q)}\rightarrow\widetilde{\mathfrak{g}}_{\pi(q)}$
given by 
\[
\ad_{[q,v]}([q,w])=\left[q,[v,w]\right].
\]
On the other hand, since the condition $\delta\Gamma\left(t\right)\in\left.\mathcal{V}\right\vert _{\Gamma\left(t\right)}$
only imposes that $\pi_{Q\ast}\left(\delta\Gamma\left(t\right)\right)=\delta\gamma\left(t\right)\in\left.C_{V}\right\vert _{\gamma\left(t\right)}$,
we have that $\delta y\left(t\right)$ and $\delta\bar{\mu}\left(t\right)$
are arbitrary and

\begin{equation}
\delta x\left(t\right)\oplus\bar{\eta}\left(t\right)=\alpha_{A}\circ p\left(\delta\gamma\left(t\right)\right)\in\alpha_{A}\circ p\left(C_{V}\right)=\mathfrak{C}_{V}.\label{dxec}
\end{equation}
As a consequence, the original variational principle \eqref{eve}
translates to the condition \eqref{rv}, with variations satisfying
\eqref{dxec}. This can be called the \textit{reduced variational
principle}. Let us study it in more detail.

\subsubsection{Generalized Hamilton-d'Alembert-Poincar\'e equations}

The arbitrariness of $\delta y\left(t\right)$ and $\delta\bar{\mu}\left(t\right)$
implies that {[}see Eq. \eqref{rv}{]}
\[
\dot{x}\left(t\right)-\frac{\partial h}{\partial y}\left(\varsigma\left(t\right)\right)=0\ \ \ \textrm{and}\ \ \ \bar{v}\left(t\right)-\frac{\partial h}{\partial\bar{\mu}}\left(\varsigma\left(t\right)\right)=0.
\]
Then, using the equality ${\displaystyle \dot{x}\left(t\right)}=x'\left(t\right)$
{[}see Eq. \eqref{def}{]}, we have that \eqref{rv} is equivalent
to 
\begin{equation}
\left\langle \frac{Dy}{Dt}\left(t\right)+\frac{\partial h}{\partial x}\left(\varsigma\left(t\right)\right)+\left\langle \bar{\mu}\left(t\right),\widetilde{B}\left(\frac{\partial h}{\partial y}\left(\varsigma\left(t\right)\right),\cdot\right)\right\rangle ,\delta x\left(t\right)\right\rangle +\left\langle \frac{D\bar{\mu}}{Dt}\left(t\right)-\ad_{\frac{\partial h}{\partial\bar{\mu}}\left(\varsigma\left(t\right)\right)}^{\ast}\bar{\mu}\left(t\right),\bar{\eta}\left(t\right)\right\rangle =0\label{hvj}
\end{equation}
and 
\begin{equation}
x'\left(t\right)=\frac{\partial h}{\partial y}\left(\varsigma\left(t\right)\right).\label{bep}
\end{equation}

So far we have seen that, if a curve $\Gamma\left(t\right)$ is a
trajectory of our GNHS, then the curve $\varsigma\left(t\right)$
given by \eqref{seda} is a solution of \eqref{hvj} and \eqref{bep},
with variations subjected to \eqref{dxec}. Reciprocally, it is easy
to show that, if $\varsigma\left(t\right)$ solves the last equations
and $\gamma\left(t\right)$ is a solution of 
\begin{equation}
\gamma'\left(t\right)=\left(\left.\alpha_{A}\circ p\right|_{\gamma\left(t\right)}\right)^{-1}\left(\varpi\left(t\right)\right),\label{recon}
\end{equation}
with 
\[
\varpi\left(t\right)=x'\left(t\right)\oplus\frac{\partial h}{\partial\bar{\mu}}\left(\varsigma\left(t\right)\right),
\]
then 
\[
\Gamma\left(t\right)\coloneqq\left(\left.\hat{\alpha}_{A}\circ\hat{p}\right|_{\gamma\left(t\right)}\right)^{-1}\left(\varsigma\left(t\right)\right)
\]
is a trajectory of our GNHS. Here, $\left.\alpha_{A}\circ p\right|_{q}$
(resp. $\left.\hat{\alpha}_{A}\circ\hat{p}\right|_{q}$) denotes the
linear isomorphism between $T_{q}Q$ (resp. $T_{q}^{*}Q$) and $T_{\pi\left(q\right)}\mathcal{X}\oplus\widetilde{\mathfrak{g}}_{\pi\left(q\right)}$
(resp. $T_{\pi\left(q\right)}^{*}\mathcal{X}\oplus\widetilde{\mathfrak{g}}_{\pi\left(q\right)}^{*}$)
described in Section \ref{gnhc}. The Eq. \eqref{recon} is precisely
a \textit{reconstruction equation}. (Notice that, in essence, it does
not depend on the system under consideration, but only on the configuration
space $Q$, the group $G$ and the chosen connection). As we said
in the introduction, we will not study in this paper the reconstruction
equations, but only the reduced ones. To continue our study of the
latter, let us assume that $A=A^{\bullet}$. This implies that {[}see
\eqref{ttt}, \eqref{sss} and \eqref{dxec}{]} 
\[
\delta x\left(t\right)\oplus\bar{\eta}\left(t\right)\in\mathfrak{C}_{V}^{\text{\textrm{hor}}}\oplus\mathfrak{C}_{V}^{\text{\textrm{ver}}}=\mathfrak{C}_{V}^{\bullet},
\]
with the reduced variations $\delta x\left(t\right)$ and $\bar{\eta}\left(t\right)$
varying independently inside 
\[
\left.\mathfrak{C}_{V}^{\text{\textrm{hor}}}\right\vert _{x\left(t\right)}\ \ \ \ \text{and\ }\ \ \ \left.\mathfrak{C}_{V}^{\text{\textrm{ver}}}\right\vert _{x\left(t\right)},
\]
respectively, what enables us to decompose Eq. \eqref{hvj} into two
parts, as we describe in the next result. 
\begin{thm}
\label{tm} Let $\left(H,D,C_{V}\right)$ be a $G$-invariant Hamiltonian
{\small{}{}GNHS} and let $A^{\bullet}:TQ\rightarrow\mathfrak{g}$
be the generalized nonholonomic connection of the system. Then, a
curve $\Gamma:\left[t_{1},t_{2}\right]\rightarrow T^{\ast}Q$ is a
trajectory of $\left(H,D,C_{V}\right)$ if and only if the curve 
\[
\varsigma:\left[t_{1},t_{2}\right]\rightarrow T^{\ast}\mathcal{X}\oplus\widetilde{\mathfrak{g}}^{\ast},\ \ \ \text{with\ base\ \ }x:\left[t_{1},t_{2}\right]\rightarrow\mathcal{X}
\]
and given by 
\[
\varsigma\left(t\right)=y\left(t\right)\oplus\bar{\mu}\left(t\right)=\hat{\alpha}_{A^{\bullet}}\circ\hat{p}\left(\Gamma\left(t\right)\right),
\]
satisfies the kinematic constraint 
\[
\varsigma\left(t\right)\in\mathfrak{D}^{\bullet},
\]
the \textbf{Horizontal Generalized Hamilton-d'Alembert-Poincar\'e (HdP)
Equations} 
\begin{equation}
\left\langle \frac{Dy}{Dt}\left(t\right)+\frac{\partial h}{\partial x}\left(\varsigma\left(t\right)\right)+\left\langle \bar{\mu}\left(t\right),\widetilde{B}\left(\frac{\partial h}{\partial y}\left(\varsigma\left(t\right)\right),\cdot\right)\right\rangle ,\delta x\left(t\right)\right\rangle =0\label{lph}
\end{equation}
and the \textbf{Vertical Generalized Hamilton-d'Alembert-Poincar\'e
(HdP) Equations} 
\begin{equation}
\left\langle \frac{D\bar{\mu}}{Dt}\left(t\right)-\ad_{\frac{\partial h}{\partial\bar{\mu}}\left(\varsigma\left(t\right)\right)}^{\ast}\bar{\mu}\left(t\right),\bar{\eta}\left(t\right)\right\rangle =0\label{lpv}
\end{equation}
for all curves 
\[
\delta x:\left[t_{1},t_{2}\right]\rightarrow T\mathcal{X}\quad\text{and}\quad\bar{\eta}:\left[t_{1},t_{2}\right]\rightarrow\widetilde{\mathfrak{g}}
\]
fulfilling 
\[
\delta x\left(t\right)\in\left.\mathfrak{C}_{V}^{\text{\textrm{hor}}}\right\vert _{x\left(t\right)}\quad\text{and\ }\quad\bar{\eta}\left(t\right)\in\left.\mathfrak{C}_{V}^{\text{\textrm{ver}}}\right\vert _{x\left(t\right)};
\]
and the base curve $x$ satisfies 
\begin{equation}
x'\left(t\right)=\frac{\partial h}{\partial y}\left(\varsigma\left(t\right)\right).\label{be}
\end{equation}
\end{thm}
This theorem can be easily proved by combining the discussion above
and Lemma 10 of Ref. \cite{cfg}.

\bigskip{}

\begin{rem}
So far, we have been dealing with a left action. For a right action,
we only have to change the sign of the Lie bracket $[v,w]$ in \eqref{rv}.
Accordingly, the term ${\displaystyle {\ad_{\frac{\partial h}{\partial\bar{\mu}}\left(\varsigma\left(t\right)\right)}^{\ast}\bar{\mu}\left(t\right)}}$
in \eqref{lpv} changes its sign and the Vertical Generalized Hamilton-d'Alembert-Poincar\'e
translates to
\[
\left\langle \frac{D\bar{\mu}}{Dt}\left(t\right)+\ad_{\frac{\partial h}{\partial\bar{\mu}}\left(\varsigma\left(t\right)\right)}^{\ast}\bar{\mu}\left(t\right),\bar{\eta}\left(t\right)\right\rangle =0.
\]
\end{rem}

Summing up, we have replaced Eq. \eqref{eve}, which, as it is well-known,
gives rise to a set of $\dim Q+\dim C_{V}$ first order ODEs, by 
\begin{itemize}
\item $\dim\mathfrak{C}_{V}^{\text{\textrm{hor}}}$ horizontal HdP equations
\eqref{lph},
\item plus $\dim\mathfrak{C}_{V}^{\text{\textrm{ver}}}$ vertical HdP equations
\eqref{lpv},
\item plus $\dim Q-\dim G$ equations for the base curve \eqref{be}, 
\end{itemize}
what gives rise to a number of 
\[
\dim Q+\dim C_{V}-\dim G
\]
first order ODEs. Thus, our reduction procedure corresponds to a reduction
of the number of equations that we must solve in order to find the
trajectories of the original GNHS (as it happens with the analogous process
for unconstrained Hamiltonian systems \cite{cmpr}).

\subsection{A reduction procedure using two connections}

\label{alt} Suppose that $\pi:Q\rightarrow\mathcal{X}$ is a trivial
bundle. In such a case, it would be desirable to take $A$ as the
related trivial connection. In fact, if we could make this choice,
then the curvature and the reduced curvature vanish, and Eq. \eqref{hvj}
would reduce to 
\[
\left\langle \frac{Dy}{Dt}\left(t\right)+\frac{\partial h}{\partial x}\left(\varsigma\left(t\right)\right),\delta x\left(t\right)\right\rangle +\left\langle \frac{D\bar{\mu}}{Dt}\left(t\right)-\ad_{\frac{\partial h}{\partial\bar{\mu}}\left(\varsigma\left(t\right)\right)}^{\ast}\bar{\mu}\left(t\right),\bar{\eta}\left(t\right)\right\rangle =0.
\]
Also, the calculation of the involved covariant derivatives are too
much easier. The problem is that the variations $\delta x$ and $\bar{\eta}$
are not independent, and we cannot decouple the above equation into
horizontal and vertical parts as we did for the $A=A^{\bullet}$ case.
(We only know that their sum must be an element of $\mathfrak{C}_{V}$).
In order to solve this problem we shall consider another reduction
procedure, which involves a second principal connection.

\subsubsection{The map $\varphi$}

\label{tc} Given a Hamiltonian GNHS $\left(H,D,C_{V}\right)$ with
symmetry, consider an arbitrary principal connection $A$ and the
generalized nonholonomic connection $A^{\bullet}$ defined from variational
constraints $C_{V}$. Consider also the isomorphisms 
\[
\alpha_{A},\alpha_{A^{\bullet}}:\left.TQ\right/G\rightarrow T\mathcal{X}\oplus\widetilde{\mathfrak{g}},
\]
and write

\[
\alpha_{A}\circ p\left(v\right)=\pi_{\ast}\left(v\right)\oplus a\left(v\right)\ \ \ \mbox{and}\ \ \alpha_{A^{\bullet}}\circ p\left(v\right)=\pi_{\ast}\left(v\right)\oplus a^{\bullet}\left(v\right).
\]
In order to avoid any confusion, given a curve $\delta\gamma$, we
shall write 
\[
\alpha_{A}\circ p\left(\delta\gamma\left(t\right)\right)=\pi_{\ast}\left(\delta\gamma\left(t\right)\right)\oplus a\left(\delta\gamma\left(t\right)\right)=\delta x\left(t\right)\oplus\bar{\eta}\left(t\right)
\]
and 
\[
\alpha_{A^{\bullet}}\circ p\left(\delta\gamma\left(t\right)\right)=\pi_{\ast}\left(\delta\gamma\left(t\right)\right)\oplus a^{\bullet}\left(\delta\gamma\left(t\right)\right)=\delta x^{\bullet}\left(t\right)\oplus\bar{\eta}^{\bullet}\left(t\right).
\]
Of course, if $\delta\gamma$ is inside $C_{V}$, then 
\[
\delta x\left(t\right)\oplus\bar{\eta}\left(t\right)\in\mathfrak{C}_{V}\quad\text{and}\quad\delta x^{\bullet}\left(t\right)\in\mathfrak{C}_{V}^{\text{\textrm{hor}}},\bar{\eta}^{\bullet}\left(t\right)\in\mathfrak{C}_{V}^{\text{\textrm{ver}}}.
\]

\bigskip{}
 In Ref. \cite{gz}, the relationship between variations $\delta x^{\bullet}$
and $\bar{\eta}^{\bullet}$ with variations $\delta x$ and $\bar{\eta}$
was found to be 
\begin{equation}
\delta x\left(t\right)=\delta x^{\bullet}\left(t\right)\ \text{and}\ \ \ \bar{\eta}\left(t\right)=\varphi\left(\delta x^{\bullet}\left(t\right)\right)+\bar{\eta}^{\bullet}\left(t\right)\label{rel}
\end{equation}
with 
\[
\varphi=P_{\widetilde{\mathfrak{g}}}\circ\alpha_{A}\circ(\alpha_{A^{\bullet}})^{-1}\circ I_{T\mathcal{X}}:T\mathcal{X}\rightarrow\widetilde{\mathfrak{g}},
\]
being 
\[
P_{\widetilde{\mathfrak{g}}}:T\mathcal{X}\oplus\widetilde{\mathfrak{g}}\rightarrow\widetilde{\mathfrak{g}}\ \ \ \text{and}\ \ \ I_{T\mathcal{X}}:T\mathcal{X}\rightarrow T\mathcal{X}\oplus\widetilde{\mathfrak{g}}
\]
the canonical projection and inclusion, respectively.\footnote{It is clear that $\varphi=0$ when $A=A^{\bullet}$.}

\subsubsection{Alternative generalized Hamilton-d'Alembert-Poincar\'e equations}

Using \eqref{rel}, Eq. \eqref{hvj} translates to the condition 
\[
\begin{array}{l}
{\displaystyle \left\langle \frac{Dy}{Dt}\left(t\right)+\frac{\partial h}{\partial x}\left(\varsigma\left(t\right)\right)+\left\langle \bar{\mu}\left(t\right),\widetilde{B}\left(\frac{\partial h}{\partial y}\left(\varsigma\left(t\right)\right),\cdot\right)\right\rangle ,\delta x^{\bullet}\left(t\right)\right\rangle +}\\
\\
{\displaystyle +\left\langle \varphi^{\ast}\left(\frac{D\bar{\mu}}{Dt}\left(t\right)-\ad_{\frac{\partial h}{\partial\bar{\mu}}\left(\varsigma\left(t\right)\right)}^{\ast}\bar{\mu}\left(t\right)\right),\delta x^{\bullet}\left(t\right)\right\rangle +}\\
\\
{\displaystyle +\left\langle \frac{D\bar{\mu}}{Dt}\left(t\right)-\ad_{\frac{\partial h}{\partial\bar{\mu}}\left(\varsigma\left(t\right)\right)}^{\ast}\bar{\mu}\left(t\right),\bar{\eta}^{\bullet}\left(t\right)\right\rangle =0}
\end{array}
\]
for all curves $\delta x^{\bullet}:\left[t_{1},t_{2}\right]\rightarrow T\mathcal{X}$
and $\bar{\eta}^{\bullet}:\left[t_{1},t_{2}\right]\rightarrow\widetilde{\mathfrak{g}}$
fulfilling 
\[
\delta x^{\bullet}\left(t\right)\in\left.\mathfrak{C}_{V}^{\text{\textrm{hor}}}\right\vert _{x\left(t\right)}\quad\text{and\ }\quad\bar{\eta}^{\bullet}\left(t\right)\in\left.\mathfrak{C}_{V}^{\text{\textrm{ver}}}\right\vert _{x\left(t\right)}.
\]
Moreover, since $\delta x^{\bullet}$ and $\bar{\eta}^{\bullet}$
are independent, we have the following result. 
\begin{thm}
Let $\left(H,D,C_{V}\right)$ be a $G$-invariant Hamiltonian {\small{}{}GNHS},
$A^{\bullet}$ its generalized nonholonomic connection and $A$ an
arbitrary principal connection. A curve $\Gamma:\left[t_{1},t_{2}\right]\rightarrow D\subset T^{\ast}Q$
is a trajectory of $\left(H,D,C_{V}\right)$ if and only if the curve
\[
\varsigma:\left[t_{1},t_{2}\right]\rightarrow T^{\ast}\mathcal{X}\oplus\widetilde{\mathfrak{g}}^{\ast},\quad\text{with base}\quad x:\left[t_{1},t_{2}\right]\rightarrow\mathcal{X}
\]
and given by 
\[
\varsigma\left(t\right)=y\left(t\right)\oplus\bar{\mu}\left(t\right)=\hat{\alpha}_{A}\circ\hat{p}\left(\Gamma\left(t\right)\right),
\]
satisfies 
\[
\varsigma\left(t\right)\in\mathfrak{D},
\]
the equations 
\begin{equation}
\begin{array}{l}
{\displaystyle \left\langle \frac{Dy}{Dt}\left(t\right)+\frac{\partial h}{\partial x}\left(\varsigma\left(t\right)\right)+\left\langle \bar{\mu}\left(t\right),\widetilde{B}\left(\frac{\partial h}{\partial y}\left(\varsigma\left(t\right)\right),\cdot\right)\right\rangle ,\delta x^{\bullet}\left(t\right)\right\rangle +}\\
\\
{\displaystyle +\left\langle \varphi^{\ast}\left(\frac{D\bar{\mu}}{Dt}\left(t\right)-\ad_{\frac{\partial h}{\partial\bar{\mu}}\left(\varsigma\left(t\right)\right)}^{\ast}\bar{\mu}\left(t\right)\right),\delta x^{\bullet}\left(t\right)\right\rangle =0}
\end{array}\label{ahph}
\end{equation}
and 
\begin{equation}
\left\langle \frac{D\bar{\mu}}{Dt}\left(t\right)-\ad_{\frac{\partial h}{\partial\bar{\mu}}\left(\varsigma\left(t\right)\right)}^{\ast}\bar{\mu}\left(t\right),\bar{\eta}^{\bullet}\left(t\right)\right\rangle =0,\label{ahpv}
\end{equation}
for all curves 
\[
\delta x^{\bullet}:\left[t_{1},t_{2}\right]\rightarrow T\mathcal{X}\quad\text{and\ }\quad\bar{\eta}^{\bullet}:\left[t_{1},t_{2}\right]\rightarrow\widetilde{\mathfrak{g}}
\]
fulfilling 
\[
\delta x^{\bullet}\left(t\right)\in\left.\mathfrak{C}_{V}^{\text{\textrm{hor}}}\right\vert _{x\left(t\right)}\quad\text{and\ }\quad\bar{\eta}^{\bullet}\left(t\right)\in\left.\mathfrak{C}_{V}^{\text{\textrm{ver}}}\right\vert _{x\left(t\right)};
\]
and the base curve $x\left(t\right)$ satisfies 
\[
x'\left(t\right)=\frac{\partial h}{\partial y}\left(\varsigma\left(t\right)\right).
\]
\end{thm}
The Theorem can be proved by combining above calculations and the
Lemma 4.6 (for the $l=0$ case) of Ref. \cite{gz}. The reduction
in the number of equations is the same as for the previous procedure.
\begin{rem}
If a right action is considered, we just have to change de sign of
${\displaystyle {\ad_{\frac{\partial h}{\partial\bar{\mu}}\left(\varsigma\left(t\right)\right)}^{\ast}\bar{\mu}\left(t\right)}}$.
\end{rem}

\begin{rem}
Note that the variables $x$, $y$ and $\bar{\mu}$ (and as a consequence
$\varsigma$), the submanifold $\mathfrak{D}$ and the curvature $B$
are related to $A$, while the variations $\delta x^{\bullet}$ and
$\bar{\eta}^{\bullet}$, and the subbundles $\mathfrak{C}_{V}^{\text{\textrm{hor}}}$
and $\mathfrak{C}_{V}^{\text{\textrm{ver}}}$, are related to $A^{\bullet}$. 
\end{rem}
Although Eqs. \eqref{ahph} and \eqref{ahpv} seem to be more complicated
than Eqs. \eqref{lph} and \eqref{lpv}, we shall see in the last
section that, for trivial principal bundles, the calculations involved
in the latter, in order to obtain the equations of motions of the
system, are substantially simpler than those involved in the former.

\section{Reduction of Hamiltonian HOCS}

The aim of this section is to extend the results of Section \ref{redgnhs},
valid for GNHS, to the case of higher order constrained systems (HOCS)
described in the Hamiltonian framework. Firstly, we shall recall the
definition of a Hamiltonian HOCS as presented in Ref. \cite{g2}.
Then, given a Lie group $G$, we shall define the idea of $G$-invariance
for these systems and develop a reduction procedure for them. Such
a procedure can be seen as a generalization of that presented in Section
\ref{alt}. First, let us introduce some notation on higher order
tangent bundles.

\bigskip{}

\textbf{Basic notation on higher-order tangent bundles. }For $k\geq0$,
let us denote by $T^{\left(k\right)}M$ the $k$\emph{-th order tangent
bundle} of $M$ (for details see \cite{LR1,Cr86}). The latter defines
a fiber bundle $\tau_{M}^{\left(k\right)}:T^{\left(k\right)}M\rightarrow M$
such that, for each $q\in M$, the fiber $T_{q}^{\left(k\right)}M$
is a set of equivalence classes $\left[\gamma\right]^{\left(k\right)}$
of curves $\gamma:\left(-\varepsilon,\varepsilon\right)\rightarrow M$
satisfying $\gamma\left(0\right)=q$. The equivalence relation is
established as follows: $\gamma_{1}\sim\gamma_{2}$ if and only if
for every local chart $\left(U,\varphi\right)$ of $M$ containing
$q$, the equations 
\[
\left.\frac{d^{s}\left(\varphi\circ\gamma_{1}\right)}{dt^{s}}\right\vert _{t=0}=\left.\frac{d^{s}\left(\varphi\circ\gamma_{2}\right)}{dt^{s}}\right\vert _{t=0},\qquad\text{for }s=0,\dots,k
\]
are fulfilled. With this definition we have the immediate identifications
$T^{\left(0\right)}M=M$ and $T^{\left(1\right)}M=TM$. Accordingly,
$\tau_{M}^{\left(0\right)}=id_{M}$ (the identity map) and $\tau_{M}^{\left(1\right)}=\tau_{M}$
(the canonical projection of $TM$ onto $M$).

Given a curve $\gamma:\left[t_{1},t_{2}\right]\rightarrow M$, its
$k$\emph{-lift }is the curve 
\[
\gamma^{\left(k\right)}:\left(t_{1},t_{2}\right)\rightarrow T^{\left(k\right)}M:\ \ t\mapsto\left[\gamma_{t}\right]^{\left(k\right)},
\]
being $\gamma_{t}:\left(-\varepsilon_{t},\varepsilon_{t}\right)\rightarrow M$
such that $\gamma_{t}\left(s\right)=\gamma\left(s+t\right)$ and $\varepsilon_{t}=\min\left\{ t-t_{1},t_{2}-t\right\} $.
If $f:N\rightarrow M$ is a smooth function, its $k$\emph{-lift}
${\displaystyle {f^{\left(k\right)}:T^{\left(k\right)}N\rightarrow T^{\left(k\right)}M}}$
is given by ${\displaystyle {f^{\left(k\right)}\left(\left[\gamma\right]^{\left(k\right)}\right)=\left[f\circ\gamma\right]^{\left(k\right)}}}$

With above identifications, the $1$-lift of a curve $\gamma$ is
precisely its velocity ${\gamma}^{\prime}:\left(t_{1},t_{2}\right)\rightarrow TM$
and $f^{\left(1\right)}=f_{\ast}$.

\subsection{Hamiltonian HOCS with symmetry}

\label{HHOCSsym} Following \cite{g,g2}, we recall the definition
of a Hamiltonian higher order constrained system. Consider a smooth
function $H:T^{\ast}Q\rightarrow\mathbb{R}$ and subsets 
\[
\mathcal{P}\subset T^{\left(k-1\right)}T^{\ast}Q\quad\text{and}\quad\mathcal{V}\subset T^{\left(l-1\right)}T^{\ast}Q\times_{T^{\ast}Q}TT^{\ast}Q
\]
with $k,l\geq1$ such that $\mathcal{P}$ is a submanifold and, for
all $\sigma\in T^{\ast}Q$ and for all $\xi\in T_{\sigma}^{\left(l-1\right)}T^{\ast}Q$,
the set 
\[
\mathcal{V}\left(\zeta\right):=\left(\left\{ \zeta\right\} \times T_{\sigma}T^{\ast}Q\right)\cap\mathcal{V},
\]
identified naturally with a subset of $T_{\sigma}T^{\ast}Q$, is either
empty or a linear subspace.\footnote{From now on, we will use such an identification and treat $\mathcal{V}\left(\zeta\right)$
as a linear subspace of $T_{\sigma}T^{\ast}Q$ without further comment.}

\begin{definition}\label{def:hocs} A Hamiltonian HOCS or simply
a HOCS is a triple $\left(H,\mathcal{P},\mathcal{V}\right)$ as given
above. We call $H$ the \textbf{Hamiltonian function}, $\mathcal{P}$
the \textbf{kinematic constraints submanifold of order} $k$ and $\mathcal{V}$
the \textbf{variational constraints subspace of order} $l$. A trajectory
of $\left(H,\mathcal{P},\mathcal{V}\right)$ is a curve $\Gamma:\left[t_{1},t_{2}\right]\rightarrow T^{\ast}Q$
such that:
\begin{enumerate}
\item $\Gamma^{\left(k-1\right)}\left(t\right)\in\mathcal{P},\quad\forall\,t\in\left(t_{1},t_{2}\right)$, 
\item the set of variations $\delta\Gamma$ of $\Gamma$ such that 
\[
\left(\Gamma^{\left(l-1\right)}\left(t\right),\delta\Gamma\left(t\right)\right)\in\mathcal{V},\quad\forall\,t\in\left(t_{1},t_{2}\right),
\]
or equivalently 
\[
\delta\Gamma\left(t\right)\in\mathcal{V}\left(\Gamma^{\left(l-1\right)}\left(t\right)\right),\quad\forall\,t\in\left(t_{1},t_{2}\right),
\]
is not empty; 
\item for all such variations the equation
\[
\int_{t_{1}}^{t_{2}}\left\langle \omega\left(\Gamma^{\prime}\left(t\right),\delta\Gamma\left(t\right)\right)-dH\left(\Gamma\left(t\right)\right),\delta\Gamma\left(t\right)\right\rangle \ dt=0
\]
must hold.
\end{enumerate}
\end{definition} 
\begin{rem}
\label{gh}Note that a GNHS $\left(H,D,\mathcal{V}\right)$ can be
seen as a HOCSs with kinematic constraints of order $k$ just defining
\[
\mathcal{P}\coloneqq\left(\tau_{T^{*}Q}^{\left(k-1\right)}\right)^{-1}\left(D\right).
\]
It is easy to see that the trajectories of $\left(H,D,\mathcal{V}\right)$
are the same as those of $\left(H,\mathcal{P},\mathcal{V}\right)$
passing through $D$. 
\end{rem}
\begin{definition} Given the space of variational constraints ${\displaystyle {\mathcal{V}\subset T^{\left(l-1\right)}T^{\ast}Q\times_{T^{\ast}Q}TT^{\ast}Q}}$,
the subset 
\[
\mathcal{W}\subset T^{\left(l-1\right)}T^{\ast}Q\times_{T^{\ast}Q}TT^{\ast}Q,
\]
defined for each $\sigma\in T^{\ast}Q$ and $\zeta\in T_{\sigma}^{\left(l-1\right)}T^{\ast}Q$
as 
\[
\mathcal{W}\left(\zeta\right):=\left(\left\{ \zeta\right\} \times T_{\sigma}T^{\ast}Q\right)\cap\mathcal{W}=\begin{cases}
\mathcal{V}^{\perp}\left(\zeta\right) & \text{if}\quad\mathcal{V}\left(\zeta\right)\neq\emptyset,\\
\emptyset & \text{otherwise},
\end{cases}
\]
is called the space of \textbf{constraint forces}. \end{definition}
Note that a Hamiltonian HOCS can also be described as the triple $\left(H,\mathcal{P},\mathcal{W}\right)$.
We are interested in HOCS such that, for all $\sigma\in T^{\ast}Q$
and $\zeta\in T_{\sigma}^{\left(l-1\right)}T^{\ast}Q$ for which $\mathcal{V}\left(\zeta\right)\neq\emptyset$,
\[
\mathcal{W}\left(\zeta\right)=\mathcal{V}^{\perp}\left(\zeta\right)\subset\ker\left(\pi_{\ast,\sigma}\right),
\]
i.e. $\mathcal{W}\left(\zeta\right)$ is a vertical subspace of $T_{\sigma}T^{\ast}Q$.
This condition is analogous to condition \textbf{A1} imposed on GNHSs
in Section \eqref{basic}: the constraint forces are given by vertical
vectors.

\bigskip{}

Now, fix an affine connection $\nabla$ on $Q$ and consider its related
isomorphism 
\[
\beta:TT^{\ast}Q\rightarrow T^{\ast}Q\oplus TQ\oplus T^{\ast}Q,
\]
given as follows. For $V\in TT^{\ast}Q$, consider a curve $u:\left(-\varepsilon,\varepsilon\right)\rightarrow T^{\ast}Q$
passing through $\tau_{T^{\ast}Q}\left(V\right)$ and with velocity
$V$ at $s=0$, i.e. $u_{\ast}\left(\left.d/ds\right\vert _{0}\right)=V$.
Then define 
\[
\beta\left(V\right):=\tau_{T^{\ast}Q}\left(V\right)\oplus\pi_{Q\ast}\left(V\right)\oplus\frac{Du}{Ds}\left(0\right),
\]
where $D/Ds$ is the covariant derivative related to $\nabla$. It
is clear that the verticality condition on $\mathcal{W}$ says that,
for all $\sigma\in T^{\ast}Q$ and all $\zeta\in T_{\sigma}^{\left(l-1\right)}T^{\ast}Q$
such that $\mathcal{V}\left(\zeta\right)\neq\emptyset$, 
\[
\beta\left(\mathcal{W}\left(\zeta\right)\right)=\sigma\oplus0\oplus F_{V}\left(\zeta\right)
\]
and (see Corollary 20, Eq. (41) on reference \cite{g2}) 
\[
\beta\left(\mathcal{V}\left(\zeta\right)\right)=\sigma\oplus C_{V}\left(\zeta\right)\oplus T_{\pi\left(\sigma\right)}^{\ast}Q,
\]
where\footnote{As is usual, $(\cdot)^{\circ}$ denote the annihilator of a vector
space.} 
\begin{equation}
F_{V}\left(\zeta\right)=\left(C_{V}\left(\zeta\right)\right)^{\circ}\label{fcv}
\end{equation}
and $C_{V}\left(\zeta\right)\subset T_{\pi\left(\sigma\right)}Q$
is a linear subspace. For later convenience, define 
\begin{equation}
C_{V}\coloneqq\bigcup_{\zeta\in T^{\left(l-1\right)}T^{\ast}Q}\left\{ \zeta\right\} \times C_{V}\left(\zeta\right)\subseteq T^{\left(l-1\right)}T^{\ast}Q\times TQ\label{eq:CV}
\end{equation}
and 
\begin{equation}
F_{V}\coloneqq\bigcup_{\zeta\in T^{\left(l-1\right)}T^{\ast}Q}\left\{ \zeta\right\} \times F_{V}\left(\zeta\right)\subseteq T^{\left(l-1\right)}T^{\ast}Q\times T^{*}Q.\label{eq:FV}
\end{equation}

\bigskip{}

The next proposition is a generalization of Eq. \eqref{eqa} for HOCS. 
\begin{prop}
Given a curve $\Gamma:\left[t_{1},t_{2}\right]\rightarrow T^{\ast}Q$,
define the curve $\gamma=$ $\pi_{Q}\circ\Gamma$. Then, a variation
$\delta\Gamma$ of $\Gamma$ satisfies $\delta\Gamma\left(t\right)\in\mathcal{V}\left(\Gamma^{\left(l-1\right)}\left(t\right)\right)$
if and only if the variation $\delta\gamma\left(t\right):=\pi_{Q\ast}\left(\delta\Gamma\left(t\right)\right)$
satisfies 
\[
\delta\gamma\left(t\right)\in C_{V}\left(\Gamma^{\left(l-1\right)}\left(t\right)\right).
\]
In other words, $\delta\Gamma$ is a variation of $\Gamma$ with values
in $\mathcal{V}$, if and only if $\delta\gamma$ is a variation of
$\gamma$ with values in $C_{V}$. 
\end{prop}
As an immediate consequence, we have the following result. 
\begin{thm}
\label{cangen} A curve $\Gamma:\left[t_{1},t_{2}\right]\rightarrow T^{\ast}Q$
is a trajectory of $\left(H,\mathcal{P},\mathcal{V}\right)$ if and
only if $\Gamma^{\left(k-1\right)}\left(t\right)\in\mathcal{P}$,
the set of variations $\delta\gamma\ $ of $\gamma=\pi_{Q}\circ\Gamma$
such that 
\[
\delta\gamma\left(t\right)\in C_{V}\left(\Gamma^{\left(l-1\right)}\left(t\right)\right)\ \ \forall t\in[t_{1},t_{2}]
\]
is not empty and for these variations 
\[
\ \gamma^{\prime}\left(t\right)=\mathbb{F}H\left(\Gamma\left(t\right)\right)\quad\text{and}\quad\left\langle \frac{D}{Dt}\Gamma\left(t\right)+\mathbb{B}H\left(\Gamma\left(t\right)\right),\delta\gamma\left(t\right)\right\rangle =0,
\]
being $\mathbb{F}H:T^{*}Q\rightarrow TQ$ and $\mathbb{B}H:T^{*}Q\rightarrow T^{*}Q$
the fiber and base derivatives of $H$, respectively. 
\end{thm}
For a proof of these two results, you may consult \cite{g2}. 
\begin{rem}
\label{vcv}Observe that, as a consequence of the last theorem, every
Hamiltonian HOCS $\left(H,\mathcal{P},\mathcal{V}\right)$ may be
described alternatively with the triple $\left(H,\mathcal{P},C_{V}\right)$
{[}see Eq. \eqref{eq:CV}{]}, and this is what we shall do from now on. 
\end{rem}
\bigskip{}

An action $\rho$ of $G$ on $Q$ gives rise to an action $\rho^{(k)}$
of $G$ on $T^{(k)}T^{*}Q$ in a canonical way. We just must consider
the $k$-lift $\hat{\rho}_{g}^{\left(k\right)}:T^{(k)}T^{*}Q\rightarrow T^{(k)}T^{*}Q$
of $\hat{\rho}_{g}$ {[}recall Eq. \eqref{rhotilde}{]} for each $g\in G$.

\begin{definition} We say that a Hamiltonian HOCS $\left(H,\mathcal{P},C_{V}\right)$
is $G$\textbf{-invariant } if for all $g\in G$ 
\begin{description}
\item [{$\mathbf{a}.$}] $H\circ\hat{\rho}_{g}=H$, 
\item [{$\mathbf{b.}$}] $\hat{\rho}_{g}^{\left(k-1\right)}\left(\mathcal{P}\right)=\mathcal{P}$, 
\item [{$\mathbf{c.}$}] for each $\sigma\in T^{*}Q$ and $\zeta\in T_{\sigma}^{\left({l-1}\right)}T^{*}Q$,
\[
\rho_{g*}\left(C_{V}(\zeta)\right)=C_{V}\left(\hat{\rho}_{g}^{(l-1)}\left(\zeta\right)\right).
\]
\end{description}
\end{definition}

Let us assume that the canonical projection $\hat{p}_{k}:T^{\left(k-1\right)}T^{*}Q\rightarrow\left.T^{\left(k-1\right)}T^{*}Q\right/G$
gives rise to a principal fiber bundle. This enable us to define the
submanifold $\mathfrak{P}\subset\left.T^{\left(k-1\right)}T^{*}Q\right/G$
given by 
\[
\mathfrak{P}:=\hat{p}_{k}\left(\mathcal{P}\right)=\mathcal{P}/G,
\]
the \textbf{reduced kinematic constraints}, and the submanifold 
\[
\mathfrak{C}_{V}\subset\left.T^{\left(l-1\right)}T^{*}Q\right/G\times_{\left.T^{*}Q\right/G}\left.TT^{*}Q\right/G,
\]
defined through the subspaces 
\[
\mathfrak{C}_{V}\left(\hat{p}_{l}\left(\zeta\right)\right):=p\left(C_{V}\left(\zeta\right)\right)=C_{V}\left(\zeta\right)/G,\quad\forall\,\zeta\in T^{\left(l-1\right)}T^{*}Q,
\]
which we shall call the \textbf{reduced variational constraints}.

\subsection{A reduction procedure}

\label{RPGNCC} Let $\left(H,\mathcal{P},C_{V}\right)$ be a $G$-invariant
Hamiltonian HOCS. As in the case of a Hamiltonian GNHSs, we will write
the equations of motion of $\left(H,\mathcal{P},C_{V}\right)$ in
terms of the reduced data $h,\mathfrak{P}$ and $\mathfrak{C}_{V}$.
Following the same reasoning as in Section \ref{alt}, we have the
next result. 
\begin{prop}
Let $\Gamma:\left[t_{1},t_{2}\right]\rightarrow T^{*}Q$ be a curve
and define 
\[
\gamma(t)=\pi_{Q}(\Gamma(t))\qquad\text{and}\qquad x(t)=\pi\left(\gamma(t)\right).
\]
If $A$ is an arbitrary principal connection, $\Gamma$ is a trajectory
of $\left(H,\mathcal{P},C_{V}\right)$ if and only if 
\[
\hat{p}_{k}\left({\Gamma}^{\left(k-1\right)}\left(t\right)\right)\in\mathfrak{P},\quad\forall\,t\in[t_{1},t_{2}],
\]
and the curve $\varsigma:\left[t_{1},t_{2}\right]\rightarrow T^{*}\mathcal{X}\oplus\widetilde{\mathfrak{g}}^{*}$
given by 
\[
\varsigma\left(t\right)=\hat{\alpha}_{A}\circ\hat{p}\left({\Gamma}\left(t\right)\right)=y\left(t\right)\oplus\bar{\mu}\left(t\right)
\]
satisfies 
\[
x'\left(t\right)=\frac{\partial h}{\partial y}\left(\varsigma\left(t\right)\right)
\]
and 
\[
\left\langle \frac{Dy}{Dt}\left(t\right)+\frac{\partial h}{\partial x}\left(\varsigma\left(t\right)\right)+\left\langle \bar{\mu}\left(t\right),\widetilde{B}\left(\frac{\partial h}{\partial y}\left(\varsigma\left(t\right)\right),\cdot\right)\right\rangle ,\delta x\left(t\right)\right\rangle +\left\langle \frac{D\bar{\mu}}{Dt}\left(t\right)-\ad_{\frac{\partial h}{\partial\bar{\mu}}\left(\varsigma\left(t\right)\right)}^{\ast}\bar{\mu}\left(t\right),\bar{\eta}\left(t\right)\right\rangle =0
\]
for all variations $\delta x(t)$ and $\bar{\eta}(t)$ such that $\delta x\left(t\right)\oplus\bar{\eta}\left(t\right)\in\mathfrak{C}_{V}\left(\hat{p}_{l}\left({\Gamma}^{(l-1)}\left(t\right)\right)\right)$. 
\end{prop}
We want to decompose the last equation into horizontal and vertical
parts as we have done for Hamiltonian GNHS. In order to do that, we
need to decompose each subspace $\mathfrak{C}_{V}\left(\hat{p}_{l}\left({\Gamma}^{(l-1)}\left(t\right)\right)\right)$.
Since these subspaces depend not only on $x\in\mathcal{X}$ but on
the points of $\left(\left.T^{\left(l-1\right)}T^{*}Q\right/G\right)_{x}$,
a standard connection is not useful in this case. We need a more general
object.

\subsubsection{The cotangent $l$-connections}

In \cite{gz}, in order to establish a reduction procedure for Lagrangian
HOCSs, the notion of an $l$-connection was presented. Analogously,
to develop a reduction for Hamiltonian HOCSs, we shall define a naturally
dual object. 
\begin{defn}
Given $l\in\mathbb{N}$, a \textbf{cotangent }$l$-\textbf{connection}
on the principal fiber bundle $\pi$ is a map 
\[
A:T^{\left(l-1\right)}T^{*}Q\times_{Q}TQ\rightarrow\mathfrak{g},
\]
such that, $\forall q\in Q$, $\forall\sigma\in T_{q}^{*}Q$ and $\forall\zeta\in T_{\sigma}^{(l-1)}T^{*}Q$,
its restriction to $\left\{ \zeta\right\} \times T_{q}Q$ is a linear
transformation and, $\forall v\in T_{q}Q$, $\forall g\in G$ and
$\forall\eta\in\mathfrak{g}$ we have that {[}compare to Eq. \eqref{cpc}{]}
\[
A\left(\zeta,\eta_{Q}\left(q\right)\right)=\eta\quad\text{and}\quad A\left(\hat{\rho}_{g}^{\left(l-1\right)}\left(\zeta\right),\rho_{g\ast}\left(v\right)\right)=\Ad_{g}A\left(\zeta,v\right).
\]
\end{defn}
\begin{rem}
Let us note that, when $l=1$, and identifying $Q\times_{Q}TQ$ with
$TQ$, we have a genuine principal connection. 
\end{rem}
From now on, and unless we state otherwise, $\sigma$ is an element
of $T_{q}^{*}Q$ for some $q\in Q$.

\begin{proposition} \label{lc}A cotangent $l$-connection is equivalent
to an assignment of a linear subspace $\mathbb{H}(\zeta)\subset T_{q}Q$
for each $\zeta\in{T_{\sigma}^{(l-1)}}T^{*}Q$ such that: 
\begin{itemize}
\item $T_{q}Q=\mathbb{H}\left(\zeta\right)\oplus\mathbb{V}\left(\zeta\right)$,
where $\mathbb{V}\left(\zeta\right)=\mathbb{V}{_{q}}=\ker\,\pi_{*,q}$, 
\item $\mathbb{H}\left(\hat{\rho}_{g}^{(l-1)}\left(\zeta\right)\right)=\rho_{g\ast}\left(\mathbb{H}\left(\zeta\right)\right)$,
$\forall\,g\in G$, and 
\item the subspaces $\mathbb{H}\left(\zeta\right)$, which we shall call
\textbf{horizontal spaces}, depend differentially on $q$ and $\zeta$. 
\end{itemize}
Given a cotangent $l$-connection $A$, the associated horizontal
spaces $\mathbb{H}\left(\zeta\right)$ are defined by 
\[
\mathbb{H}\left(\zeta\right)=\left\{ v\in T_{q}Q:A\left(\zeta,v\right)=0\right\} .
\]
Reciprocally, given horizontal spaces $\mathbb{H}\left(\zeta\right)$
satisfying the properties listed above, the corresponding cotangent
$l$-connection $A$ is defined by the formula 
\[
A\left(\zeta,v\right)=\eta,
\]
where $\eta\in\mathfrak{g}$ is such that $v-\eta_{Q}\left(q\right)\in\mathbb{H}\left(\zeta\right)$.
\end{proposition}

(For a proof, see Ref. \cite{gz}). Related to a cotangent $l$-connection
we have a map 
\begin{equation}
\alpha_{A}:\left.T^{\left(l-1\right)}T^{*}Q/G\times_{\mathcal{X}}TQ\right/G\rightarrow T\mathcal{X}\oplus\widetilde{\mathfrak{g}},\label{aah}
\end{equation}
similar to the Atiyah isomorphism of a principal connection, defined
in the following way: 
\begin{enumerate}
\item Take $[\zeta]\in\left(T^{\left(l-1\right)}T^{*}Q\right)/G$ and $[v]\in TQ/G$,
both of them based on the same point $x\in\mathcal{X}$. 
\item Consider representatives $\zeta\in T_{\sigma}^{\left(l-1\right)}T^{*}Q$
and $v\in T_{q}Q$ of each one of these classes, such that $\pi\left(q\right)=x$
(observe that this is always possible). 
\item Then, define 
\[
\alpha_{A}\left(\left[\zeta\right],\left[v\right]\right)\coloneqq\pi_{\ast}\left(v\right)\oplus\left[q,A\left(\zeta,v\right)\right].
\]
\end{enumerate}
Following \cite{gz}, we can see that $\alpha_{A}$ is well defined.
Besides, we can prove that, for each $\zeta\in T^{(l-1)}T^{*}Q$,
the map 
\[
\alpha_{A}^{\left[\zeta\right]}:\left(\left.TQ\right/G\right)_{\pi\left(q\right)}\rightarrow T_{\pi\left(q\right)}\mathcal{X}\oplus\widetilde{\mathfrak{g}}_{\pi\left(q\right)},
\]
given by 
\begin{equation}
\alpha_{A}^{\left[\zeta\right]}\left(\left[v\right]\right)\coloneqq\alpha_{A}\left(\left[\zeta\right],\left[v\right]\right),\label{aah2}
\end{equation}
defines a linear isomorphism.

\bigskip{}

For later convenience, let us define the map $a:T^{\left(l-1\right)}T^{*}Q\times_{Q}TQ\rightarrow\widetilde{\mathfrak{g}}$
such that 
\[
a\left(\zeta,v\right)\coloneqq\left[q,A\left(\zeta,v\right)\right],
\]
and the maps $a_{\zeta}:TQ\rightarrow\widetilde{\mathfrak{g}}$ given
by 
\[
a_{\zeta}\left(v\right)\coloneqq\left[q,A\left(\zeta,v\right)\right].
\]
It follows that 
\begin{equation}
\alpha_{A}\left(\left[\zeta\right],\left[v\right]\right)=\pi_{\ast}\left(v\right)\oplus a\left(\zeta,v\right),\label{aapa}
\end{equation}
where $\zeta$ and $v$ are representatives based on the same point
$q$.

\subsubsection{The higher order cotangent connection}

\label{hocc}

In this subsection we shall see that to each $G$-invariant Hamiltonian
HOCS a particular cotangent $l$-connection can be assigned. It will
be called \emph{higher order cotangent connection}, and it will enable
us to separate the reduced virtual displacements $\mathfrak{C}_{V}$
into horizontal and vertical components. The construction of such
an object will be done in several steps (compare with the higher order
$l$-connection appearing in \cite{gz}). 
\begin{enumerate}
\item Fix a $G$-invariant metric on $Q$. We shall assume that $H$ is
simple, and that we choose the Riemannian metric defining its kinetic
term. 
\item For each $q\in Q$, $\sigma\in T_{q}^{*}Q$ and $\zeta\in T_{\sigma}^{(l-1)}T^{*}Q$,
consider 
\[
\mathcal{S}\left(\zeta\right):=C_{V}\left(\zeta\right)\cap\mathbb{V}\left(\zeta\right)
\]
and write 
\[
C_{V}\left(\zeta\right)=\mathcal{T}\left(\zeta\right)\oplus\mathcal{S}\left(\zeta\right)\ \ \ \ \text{and}\ \ \ \mathbb{V}\left(\zeta\right)=\mathcal{S}\left(\zeta\right)\oplus\mathcal{U}\left(\zeta\right),
\]
where $\mathcal{T}\left(\zeta\right)$ and $\mathcal{U}\left(\zeta\right)$
are the orthogonal complements of $\mathcal{S}\left(\zeta\right)$
in $C_{V}\left(\zeta\right)$ and $\mathbb{V}\left(\zeta\right)$,
respectively. Recall that $\mathbb{V}(\zeta)=\mathbb{V}_{q}$ is the
vertical space at $q$ associated to $\pi$. 
\item Consider the orthogonal complement of $C_{V}\left(\zeta\right)+\mathbb{V}\left(\zeta\right)$
in $T_{q}Q$. Let us denote it $\mathcal{R}\left(\zeta\right)$.

We shall assume that the spaces $\mathcal{R}\left(\zeta\right)\oplus\mathcal{T}\left(\zeta\right)$
depend differentially on $q$ and $\zeta$. 
\item Define \textbf{higher order cotangent} $l$-\textbf{connection} $A^{\bullet}:T^{\left(l-1\right)}T^{*}Q\times_{Q}TQ\rightarrow\mathfrak{g}$, with horizontal subspaces (see Proposition \ref{lc}) 
\[
\mathbb{H}^{\bullet}\left(\zeta\right):=\mathcal{R}\left(\zeta\right)\oplus\mathcal{T}\left(\zeta\right).
\]
In other words, given $v\in T_{q}Q$, define 
\[
A^{\bullet}\left(\zeta,v\right)=\eta
\]
if $v-\eta_{Q}\left(q\right)\in\mathbb{H}^{\bullet}\left(\zeta\right)$. 
\end{enumerate}
It is easy to show that $A^{\bullet}$ is effectively a cotangent
$l$-connection. In particular, 
\[
T_{q}Q=\mathbb{H}^{\bullet}\left(\zeta\right)\oplus\mathbb{V}_{q}.
\]
Note that 
\[
\mathcal{T}\left(\zeta\right)=C_{V}\left(\zeta\right)\cap\mathbb{H}^{\bullet}\left(\zeta\right).
\]
Thus, 
\begin{equation}
C_{V}\left(\zeta\right)=\left[C_{V}\left(\zeta\right)\cap\mathbb{H}^{\bullet}\left(\zeta\right)\right]\oplus\left[C_{V}\left(\zeta\right)\cap\mathbb{V}_{q}\right].\label{decocv}
\end{equation}

Using the isomorphisms 
\[
\alpha_{A^{\bullet}}^{\left[\zeta\right]}:\left(\left.TQ\right/G\right)_{\pi\left(q\right)}\rightarrow T_{\pi\left(q\right)}\mathcal{X}\oplus\widetilde{\mathfrak{g}}_{\pi\left(q\right)}
\]
and Eqs. \eqref{aah} and \eqref{aah2}, we have 
\begin{equation}
\mathbb{H}^{\bullet}\left(\zeta\right)/G\simeq\alpha_{A^{\bullet}}^{\left[\zeta\right]}\left(\mathbb{H}^{\bullet}\left(\zeta\right)/G\right)=\pi_{\ast}\left(\mathbb{H}^{\bullet}\left(\zeta\right)\right)=T_{\pi\left(q\right)}\mathcal{X}\label{htx}
\end{equation}
and {[}see Eq. \eqref{aapa}{]} 
\begin{equation}
\mathbb{V}_{q}/G\simeq\alpha_{A^{\bullet}}^{\left[\zeta\right]}\left(\mathbb{V}_{q}/G\right)=a_{\zeta}^{\bullet}\left(\mathbb{V}_{q}\right)=\widetilde{\mathfrak{g}}_{\pi\left(q\right)}.\label{vg}
\end{equation}
Accordingly, combining \eqref{decocv}, \eqref{htx} and \eqref{vg},
the next result is immediate. 
\begin{prop}
\label{decoh} If we define $\mathfrak{C}_{V}^{\bullet}\left(\left[\zeta\right]\right):=\alpha_{A^{\bullet}}^{\left[\zeta\right]}\circ p\left(C_{V}\left(\zeta\right)\right)$,
we have 
\[
\mathfrak{C}_{V}^{\bullet}\left(\left[\zeta\right]\right)=\mathfrak{C}_{V}^{\text{\textrm{hor}}}\left(\left[\zeta\right]\right)\oplus\mathfrak{C}_{V}^{\text{\textrm{ver}}}\left(\left[\zeta\right]\right)
\]
where 
\[
\mathfrak{C}_{V}^{\text{\textrm{hor}}}\left(\left[\zeta\right]\right)\simeq\pi_{\ast}\left(C_{V}\left(\zeta\right)\right)=T_{\pi\left(q\right)}\mathcal{X}\cap\mathfrak{C}_{V}^{\bullet}\left(\left[\zeta\right]\right)
\]
and 
\[
\mathfrak{C}_{V}^{\text{\textrm{ver}}}\left(\left[\zeta\right]\right)\simeq a_{\zeta}^{\bullet}\left(C_{V}\left(\zeta\right)\right)=\widetilde{\mathfrak{g}}_{\pi\left(q\right)}\cap\mathfrak{C}_{V}^{\bullet}\left(\left[\zeta\right]\right).
\]
\end{prop}

\subsubsection{The maps $\varphi^{\left[\zeta\right]}$}

Let us relate the description of $p\left(C_{V}\left(\zeta\right)\right)$
via $A^{\bullet}$ and an arbitrary connection $A$. Consider a curve
$\Gamma:(t_{1},t_{2})\rightarrow T^{*}Q$ and the projected curve
on $Q$ given by $\gamma(t)=\pi_{Q}(\Gamma(t))$. If $\delta\gamma$
denotes an infinitesimal variation on $\gamma$, let us write 
\[
\alpha_{A}\circ p\left(\delta\gamma\left(t\right)\right)=\pi_{\ast}\left(\delta\gamma\left(t\right)\right)\oplus a\left(\delta\gamma\left(t\right)\right)=\delta x\left(t\right)\oplus\bar{\eta}\left(t\right)
\]
as before, and 
\[
\alpha_{A^{\bullet}}^{\left[\zeta\right]}\circ p\left(\delta\gamma\left(t\right)\right)=\pi_{\ast}\left(\delta\gamma\left(t\right)\right)\oplus a_{\zeta\left(t\right)}^{\bullet}\left(\delta\gamma\left(t\right)\right)=\delta x^{\bullet}\left(t\right)\oplus\bar{\eta}^{\bullet}\left(t\right),
\]
where $\zeta\left(t\right)=\Gamma^{\left(l-1\right)}\left(t\right)$.
It is clear that, if $\delta\gamma\left(t\right)\in C_{V}\left(\Gamma^{\left(l-1\right)}\left(t\right)\right)$
then
\[
\delta x\left(t\right)\oplus\bar{\eta}\left(t\right)\in\mathfrak{C}_{V}\left(\left[\Gamma^{\left(l-1\right)}\left(t\right)\right]\right)
\]
and 
\[
\delta x^{\bullet}\left(t\right)\in\mathfrak{C}_{V}^{\text{\textrm{hor}}}\left(\left[\Gamma^{\left(l-1\right)}\left(t\right)\right]\right)\ \ \ \text{and}\ \ \ \bar{\eta}^{\bullet}\left(t\right)\in\mathfrak{C}_{V}^{\text{\textrm{ver}}}\left(\left[\Gamma^{\left(l-1\right)}\left(t\right)\right]\right).
\]
By using Proposition \ref{decoh}, all the reduced variations inside
$\mathfrak{C}_{V}$ can be written in terms of independent variations
$\delta x^{\bullet}\in\mathfrak{C}_{V}^{\text{\textrm{hor}}}$ and
$\bar{\eta}^{\bullet}\in\mathfrak{C}_{V}^{\text{\textrm{ver}}}$.
As we noticed in Section \ref{tc}, we can write expressions for the
variations $\delta x$ and $\bar{\eta}$ in terms of $\delta x^{\bullet}$,
$\bar{\eta}^{\bullet}$ and the canonical projections as follows 
\[
\delta x\left(t\right)=\delta x^{\bullet}\left(t\right),\ \ \ \bar{\eta}\left(t\right)=\varphi^{\left[\zeta\right]}\left(\delta x^{\bullet}\left(t\right)\right)+\bar{\eta}^{\bullet}\left(t\right),
\]
where $\varphi^{\left[\zeta\right]}:T_{\pi\left(q\right)}\mathcal{X}\rightarrow\widetilde{\mathfrak{g}}_{\pi\left(q\right)}$
is given by 
\[
\varphi^{\left[\zeta\right]}\left(u\right)\coloneqq P_{\widetilde{\mathfrak{g}}}\circ\alpha_{A}\circ\left(\alpha_{A^{\bullet}}^{\left[\zeta\right]}\right)^{-1}\circ I_{T\mathcal{X}}\left(u\right).
\]
Observe that $\varphi^{\left[\zeta\right]}$ gives rise to another
map 
\[
\varphi:T^{(l-1)}T^{*}Q/G\times_{\mathcal{X}}T\mathcal{X}\rightarrow\widetilde{\mathfrak{g}}
\]
defined by 
\[
\varphi\left(\left[\zeta\right],u\right)=\varphi^{\left[\zeta\right]}\left(u\right),\quad\forall\,q\in Q,\;\sigma\in T_{q}^{*}Q,\;\zeta\in T_{\sigma}^{\left(l-1\right)}T^{*}Q\quad\text{and}\quad u\in T_{\pi\left(q\right)}\mathcal{X}.
\]

\subsubsection{The Higher Order Hamilton-d'Alembert-Poincar\'e (HdP) equations}

We shall finally derive a set of equations describing the dynamics
of a $G$-invariant Hamiltonian {\small{}{}HOCS} $\left(H,\mathcal{P},C_{V}\right)$
in terms of their corresponding reduced variables on $T^{*}\mathcal{X}$
and $\widetilde{\mathfrak{g}}^{*}$. In order to write these equations,
we shall prove that the fiber bundles 
\[
T^{(n)}T^{*}Q/G\quad\text{and}\quad T^{(n)}T^{*}\mathcal{X}\times_{\mathcal{X}}\left[(n+1)\widetilde{\mathfrak{g}}^{*}\oplus n\widetilde{\mathfrak{g}}\right]
\]
are isomorphic. In the first place, we need the next result. 
\begin{lem}
If $A:TQ\rightarrow\mathfrak{g}$ is a principal connection on $\pi:Q\rightarrow\mathcal{X}$,
then $\hat{A}\coloneqq A\circ\pi_{Q*}:TT^{*}Q\rightarrow\mathfrak{g}$
is a principal connection on $\hat{p}:T^{*}Q\rightarrow T^{*}Q/G$. 
\end{lem}
\begin{proof}
It is clear that the function $\hat{A}:TT^{*}Q\rightarrow\mathfrak{g}$,
defined as 
\[
\hat{A}(v_{\sigma_{q}})=A\left(\pi_{Q*}(v_{\sigma_{q}})\right),
\]
is linear and, for all $\eta\in\mathfrak{g}$ and all $g\in G$, it satisfies
\[
\hat{A}(\eta_{T^{*}Q}(\sigma_{q}))=A(\pi_{Q*}(\eta_{T^{*}Q}(\sigma_{q}))=A(\eta_{Q}(q))=\eta,
\]
and

\begin{align*}
\hat{A}\left(\hat{\rho}_{g}^{(1)}(v_{\sigma_{q}})\right) & =A\left(\pi_{Q*}\left(\hat{\rho}_{g}^{(1)}(v_{\sigma_{q}})\right)\right)=A\left(\rho_{g*}\left(\pi_{Q*}(v_{\sigma_{q}})\right)\right)\\
 & =\Ad_{g^{-1}}A(\pi_{Q*}(v_{\sigma_{q}}))=\Ad_{g^{-1}}\hat{A}(v_{\sigma_{q}}).
\end{align*}
So, $\hat{A}$ is indeed a principal connection on $\hat{p}$. 
\end{proof}
Now, let us consider the fiber bundle $\hat{\mathfrak{g}}:=\left(T^{*}Q\times\mathfrak{g}\right)/G$
where $G$ acts on $T^{*}Q$ (resp. on $\mathfrak{g}$) through the
canonical lifted (resp. adjoint) action. Observe that this bundle
is the adjoint bundle of $\hat{p}$ with base $T^{*}Q/G$. Its elements
will be denoted by $[\sigma,\eta]$, where $\sigma\in T^{*}Q$ and
$\eta\in\mathfrak{g}$.

It is clear that $\hat{A}$ gives rise to the isomorphism 
\[
\alpha_{\hat{A}}:TT^{*}Q/G\rightarrow T\left(T^{*}Q/G\right)\oplus\hat{\mathfrak{g}}
\]
given by 
\[
\alpha_{\hat{A}}\left(\left[v_{\sigma_{q}}\right]\right)=\hat{p}_{*}(v_{\sigma_{q}})\oplus\left[\sigma_{q},\hat{A}\left(v_{\sigma_{q}}\right)\right],\quad\forall\,v_{\sigma_{q}}\in TT^{*}Q.
\]
Denoting by $\hat{a}$ the map 
\[
\hat{a}:TT^{*}Q\rightarrow\hat{\mathfrak{g}}\ :\ v_{\sigma_{q}}\mapsto\left[\sigma_{q},\hat{A}\left(v_{\sigma_{q}}\right)\right],
\]
we have that 
\[
\alpha_{\hat{A}}\left(\left[v_{\sigma_{q}}\right]\right)=\hat{p}\left(v_{\sigma_{q}}\right)\oplus\hat{a}\left(v_{\sigma_{q}}\right).
\]
Moreover, according to Reference \cite{cmr0}, related to $\hat{A}$
we have the following isomorphisms. 
\begin{lem}
\label{l1}For each $n\geq1$, we have a bundle isomorphism 
\[
\alpha_{\hat{A}}^{(n)}:T^{(n)}T^{*}Q/G\rightarrow T^{(n)}\left(T^{*}Q/G\right)\oplus n\hat{\mathfrak{g}},
\]
where $n\hat{\mathfrak{g}}$ denotes the Whitney sum of $n$ copies
of $\hat{\mathfrak{g}}$. For a curve $\Gamma:[t_{1},t_{2}]\rightarrow T^{*}Q$
, this isomorphism is given by 
\[
\alpha_{\hat{A}}^{(n)}\left(\left[\Gamma^{(n)}(t)\right]\right)=\left([\hat{p}\circ\Gamma]^{(n)}(t),\oplus_{i=0}^{n-1}\frac{D^{i}\hat{a}\left(\Gamma'(t)\right)}{Dt^{i}}\right)
\]
where ${\displaystyle {\frac{D^{i}\hat{a}\left(\Gamma'(t)\right)}{Dt^{i}}}}$
denotes the $i$-th covariant derivative of the curve $\hat{a}\left(\Gamma'(t)\right)$
in $\hat{\mathfrak{g}}$. 
\end{lem}
Consider the maps $a$ and $\hat{a}$ related to the connections $A$
and $\hat{A}$. 
\begin{lem}
There exists a fiber bundle morphism $\mathfrak{p}:\hat{\mathfrak{g}}\rightarrow\widetilde{\mathfrak{g}}$
such that 
\[
\mathfrak{p}\circ\hat{a}=a\circ\pi_{Q*}.
\]
Moreover, the map 
\[
Id_{T^{(n)}\left(T^{*}Q/G\right)}\times\mathfrak{p}:T^{(n)}\left(T^{*}Q/G\right)\oplus_{T^{*}Q/G}\hat{\mathfrak{g}}\rightarrow T^{(n)}\left(T^{*}Q/G\right)\oplus_{\mathcal{X}}\widetilde{\mathfrak{g}}
\]
is a fiber bundle isomorphism. 
\end{lem}
\begin{proof}
First, let us define $\mathfrak{p}:\hat{\mathfrak{g}}\rightarrow\widetilde{\mathfrak{g}}$
as ${\displaystyle {\mathfrak{p}([\sigma,\eta])=[\pi_{Q}(\sigma),\eta]}}$.
Using Eq. \eqref{lad}, it easily follows that $\mathfrak{p}$ is
well defined. It is clear that this function is a fiber bundle morphism
between the fiber bundles $\hat{\mathfrak{g}}$ and $\widetilde{\mathfrak{g}}$
over the quotient map $[\pi_Q]:T^*Q/G\rightarrow \mathcal X$ sucht that $[\pi_Q]([\alpha])=[\pi_Q(\alpha)]$ for all $[\alpha]\in T^*Q/G$. Also, 
\[
\mathfrak{p}(\hat{a}(v_{\sigma_{q}}))=\mathfrak{p}([\sigma_{q},\hat{A}(v_{\sigma_{q}})])=[\pi_{Q}(\sigma_{q}),\hat{A}(v_{\sigma_{q}})]=[\pi_{Q}(\sigma_{q}),A(\pi_{Q*}(v_{\sigma_{q}}))]=a(\pi_{Q*}(v_{\sigma_{q}})),
\]
what implies that $\mathfrak{p}\circ\hat{a}=a\circ\pi_{Q*}$. Finally,
given a curve $\zeta:\left(-\varepsilon,\varepsilon\right)\rightarrow T^{*}Q$,
define the map 
\[
\left([\zeta]^{(n)},[q,\eta]\right)\in T^{(n)}\left(T^{*}Q/G\right)\oplus_{\mathcal{X}}\widetilde{\mathfrak{g}}\mapsto\left([\zeta]^{(n)},[\zeta_{0},\eta]\right)\in T^{(n)}\left(T^{*}Q/G\right)\oplus_{T^{*}Q/G}\hat{\mathfrak{g}},
\]
where $\zeta_{0}=\zeta(0)\in T^{*}Q$ and $(q,\eta)\in Q\times\mathfrak{g}$
is a representative such that\footnote{Notice that this is always possible.}
$\pi_{Q}(\zeta_{0})=q$. A straightforward computation shows that
such a map is well defined and is a smooth inverse of $Id_{\left(T^{(n)}T^{*}Q\right)/G}\times\mathfrak{p}$. 
\end{proof}
As an immediate consequence we have the following. 
\begin{cor}
\label{coro1}The map $Id_{T^{(n)}\left(\left(T^{*}Q\right)/G\right)}\times m\mathfrak{p}$
is a fiber bundle isomorphism between the spaces 
\[
T^{(n)}\left(T^{*}Q/G\right)\oplus_{T^{*}Q/G}m\hat{\mathfrak{g}}\qquad\text{and}\qquad T^{(n)}\left(T^{*}Q/G\right)\oplus_{\mathcal{X}}m\widetilde{\mathfrak{g}}
\]
for every $n,m\geq1$. 
\end{cor}
Now, consider the next result on general vector bundles. 
\begin{lem}
\label{l2}Given a vector bundle $\Pi:V\rightarrow Y$ and an affine
connection on it, there exists an isomorphism between the fiber bundles
\[
T^{(n)}V\quad\text{and}\quad T^{(n)}Y\times_{Y}(n+1)\,V.
\]
And given a second vector vector bundle $W\rightarrow Y$ and an affine
connection on it, we have the isomorphisms 
\[
T^{(n)}\left(V\oplus W\right)\simeq T^{(n)}Y\times_{Y}(n+1)\,V\times_{Y}(n+1)\,W
\]
and 
\[
T^{(n)}\left(V\oplus W\right)\simeq T^{(n)}V\times_{Y}(n+1)\,W.
\]
\end{lem}
\begin{proof}
It is enough to show the first statement. Given a curve $\Gamma:\left(-\varepsilon,\varepsilon\right)\rightarrow V$,
a possible isomorphism is given by the assignment 
\[
\left[\Gamma\right]^{\left(n\right)}\mapsto\left(\left[\Pi\circ\Gamma\right]^{\left(n\right)},\oplus_{i=0}^{n}\frac{D^{i}\Gamma}{Dt^{i}}\left(0\right)\right).
\]
The details are left to the reader. 
\end{proof}
Using above identifications, we can prove the wanted isomorphism. 
\begin{thm}
Any principal connection $A$ on $\pi$ gives rise to an isomorphism
between the fiber bundles 
\[
T^{(n)}T^{*}Q/G\quad\text{and}\quad T^{(n)}T^{*}\mathcal{X}\times_{\mathcal{X}}\left[(n+1)\widetilde{\mathfrak{g}}^{*}\oplus n\widetilde{\mathfrak{g}}\right].
\]
\end{thm}
\begin{proof}
Combining Lemma \ref{l1} and Corollary \ref{coro1}, for any $n\in\mathbb{N}$,
we have the fiber bundle isomorphism 
\[
(Id_{T^{(n)}\left(T^{*}Q/G\right)}\times n\mathfrak{p})\circ\alpha_{\hat{A}}^{(n)}:\left(T^{(n)}T^{*}Q\right)/G\rightarrow T^{(n)}\left(T^{*}Q/G\right)\oplus_{\mathcal{}}n\widetilde{\mathfrak{g}}.
\]
On the other hand, using the $n$-lift 
\[
\left(\hat{\alpha}_{A}\right){}^{(n)}:T^{(n)}\left(T^{*}Q/G\right)\rightarrow T^{(n)}\left(T^{*}\mathcal{X}\oplus\widetilde{\mathfrak{g}}^{*}\right)
\]
of $\hat{\alpha}_{A}:T^{*}Q/G\rightarrow T^{*}\mathcal{X}\oplus\widetilde{\mathfrak{g}}^{*}$
and the third equation in Lemma \ref{l2}, it is immediate that fixing
an affine connection on $T^{*}\mathcal{X}\oplus\widetilde{\mathfrak{g}}^{*}$
we can construct an isomorphism between $T^{(n)}\left(T^{*}Q/G\right)$
and 
\[
T^{(n)}T^{*}\mathcal{X}\times_{\mathcal{X}}(n+1)\,\widetilde{\mathfrak{g}}^{*}.
\]
Composing the above mentioned isomorphisms, the theorem follows. 
\end{proof}
\begin{rem}
Given a curve $\Gamma:[t_{1},t_{2}]\rightarrow T^{*}Q$, consider
as in Section \ref{RPGNCC} the curve $\zeta:[t_{1},t_{2}]\rightarrow T^{*}\mathcal{X}\oplus\widetilde{\mathfrak{g}}^{*}$,
given by 
\[
\zeta(t)\coloneqq\hat{\alpha}_{A}\circ\hat{p}(\Gamma(t))\eqqcolon y(t)\oplus\bar{\mu}(t).
\]
Consider also the curves 
\[
\gamma(t)=\pi_{Q}(\Gamma(t))\quad\text{and}\quad x(t)=\pi(\gamma(t)),
\]
and 
\[
\dot{x}(t)\oplus\bar{v}(t)\coloneqq\alpha_{A}([\gamma'(t)]).
\]
Also, consider on $T^{*}\mathcal{X}\oplus\widetilde{\mathfrak{g}}^{*}$
an affine connection $\nabla=\nabla^{\mathcal{X}}\oplus\nabla^{A}$.

It is worth mentioning that the covariant derivative of a curve on
$\widetilde{\mathfrak{g}}^{*}$ with respect to the affine connection
$\nabla^{A}$ coincides, by definition, with the covariant derivative
with respect to the principal connection $A$.

Then, the isomorphism constructed in the proof of the last theorem
is given by 
\[
[\Gamma^{(n)}(t)]\mapsto\left(y^{(n)}(t),\oplus_{i=0}^{n}\frac{D^{i}\bar{\mu}}{Dt^{i}}(t),\oplus_{i=0}^{n-1}\frac{D^{i}\bar{v}}{Dt^{i}}(t)\right),
\]
where $\frac{D^{i}\bar{\mu}}{Dt^{i}}(t)$ (resp. $\frac{D^{i}\bar{v}}{Dt^{i}}(t)$)
denote the $i$-th covariant derivative of curves on $\mathfrak{\widetilde{g}^{*}}$
(resp. $\mathfrak{\widetilde{g}}$) with respect to the affine connection
$\nabla^{A}$ (resp. the principal connection $A$).
\end{rem}
Now, we are able to write down the desired equations. 
\begin{thm}
Let $\left(H,\mathcal{P},C_{V}\right)$ be a $G$-invariant Hamiltonian
{\small{}{}HOCS} and let us denote by $A^{\bullet}$ its associated
higher order cotangent $l$-connection and $A$ an arbitrary principal
connection on $\pi$. A curve $\Gamma:\left[t_{1},t_{2}\right]\rightarrow T^{*}Q$
is a trajectory of $\left(H,\mathcal{P},C_{V}\right)$ if and only
if the curve 
\[
\varsigma:\left[t_{1},t_{2}\right]\rightarrow T^{*}\mathcal{X}\oplus\widetilde{\mathfrak{g}}^{*},
\]
given by 
\[
\varsigma\left(t\right)=\hat{\alpha}_{A}\circ\hat{p}\left({\Gamma}\left(t\right)\right)=y\left(t\right)\oplus\bar{\mu}\left(t\right),
\]
satisfies 
\[
\left(y^{(k)}(t),\oplus_{i=0}^{k}\frac{D^{i}\bar{\mu}}{Dt^{i}}(t),\oplus_{i=0}^{k-1}\frac{D^{i}}{Dt^{i}}\left(\frac{\partial h}{\partial\bar{\mu}}\left(\varsigma\left(t\right)\right)\right)\right)\in\mathfrak{P},
\]
\textbf{the Higher Order HdP Horizontal Equations} 
\begin{align*}
 & \left\langle \frac{Dy}{Dt}\left(t\right)+\frac{\partial h}{\partial x}\left(\varsigma\left(t\right)\right)+\left\langle \bar{\mu}\left(t\right),\widetilde{B}\left(\frac{\partial h}{\partial y}\left(\varsigma\left(t\right)\right),\cdot\right)\right\rangle ,\delta x^{\bullet}\left(t\right)\right\rangle \nonumber \\
 & +\left\langle \left(\varphi^{\mathfrak{c}(t)}\right)^{*}\left(\frac{D\bar{\mu}}{Dt}\left(t\right)-\ad_{\frac{\partial h}{\partial\bar{\mu}}\left(\varsigma\left(t\right)\right)}^{\ast}\bar{\mu}\left(t\right)\right),\delta x^{\bullet}\left(t\right)\right\rangle =0\label{hdph}
\end{align*}
and \textbf{the Higher Order HdP Vertical Equations} 
\begin{equation*}
\left\langle \frac{D\bar{\mu}}{Dt}\left(t\right)-\ad_{\frac{\partial h}{\partial\bar{\mu}}\left(\varsigma\left(t\right)\right)}^{\ast}\bar{\mu}\left(t\right),\bar{\eta}^{\bullet}\left(t\right)\right\rangle =0\label{hdpv}
\end{equation*}
for all curves 
\[
\delta x^{\bullet}:\left[t_{1},t_{2}\right]\rightarrow T\mathcal{X}\quad\text{and}\quad\bar{\eta}^{\bullet}:\left[t_{1},t_{2}\right]\rightarrow\widetilde{\mathfrak{g}}
\]
satisfying 
\[
\delta x^{\bullet}\left(t\right)\in\mathfrak{C}_{V}^{\text{\textrm{hor}}}\left(\mathfrak{c}\left(t\right)\right)\quad\text{and}\quad\bar{\eta}^{\bullet}\left(t\right)\in\mathfrak{C}_{V}^{\text{\textrm{ver}}}\left(\mathfrak{c}\left(t\right)\right)
\]
where 
\[
\mathfrak{c}\left(t\right)=\left(y^{(l)}(t),\oplus_{i=0}^{l}\frac{D^{i}\bar{\mu}}{Dt^{i}}(t),\oplus_{i=0}^{l-1}\frac{D^{i}}{Dt^{i}}\left(\frac{\partial h}{\partial\bar{\mu}}\left(\varsigma\left(t\right)\right)\right)\right),
\]
and the base curve $x(t)$ satisfies 
\[
x'\left(t\right)=\frac{\partial h}{\partial y}\left(\varsigma\left(t\right)\right).
\]
\end{thm}
The theorem can be proved by combining our previous results and the
proof of Lemma 4.6 of Ref. \cite{gz}. Regarding the number of reduced
equations, we have the same as for the case of GNHSs.
\begin{rem}
For a right action, recall that we have to change the sign of ${\displaystyle {\ad_{\frac{\partial h}{\partial\bar{\mu}}\left(\varsigma\left(t\right)\right)}^{\ast}\bar{\mu}\left(t\right)}}$.
\end{rem}
\begin{rem}
The variables $x$, $y$ and $\bar{\mu}$, the submanifold $\mathfrak{P}$,
the curvature $B$ and the curve $\mathfrak{c}\left(t\right)$ are
related to $A$, while the variations $\delta x^{\bullet}$ and $\bar{\eta}^{\bullet}$,
and subspaces $\mathfrak{C}_{V}^{\text{\textrm{hor}}}\left(\mathfrak{c}\left(t\right)\right)$
and $\mathfrak{C}_{V}^{\text{\textrm{ver}}}\left(\mathfrak{c}\left(t\right)\right)$,
are related to $A^{\bullet}$. 
\end{rem}

\section{The case of trivial bundles}

\label{trivial}

In this section we study the form that reduced equations obtained
in the previous sections adopt when the configuration space $Q$ of
the system is a trivial principal bundle with structure group $G$.
In the first place we shall focus on GNHSs. At the end of the section
we briefly explain how to deal with the case of HOCSs.

\bigskip{}

We describe the reduction procedure using right actions, instead of
left ones, to emphasize that all the computations are analogous for
both kind of actions.

\bigskip{}

Given a manifold $\mathcal{X}$ and a Lie group $G$, let us consider
the product manifold $Q=\mathcal{X}\times G$ and the \textbf{right
action of} $G$ on $Q$ induced by right translation of $G$, i.e.\footnote{We denote by $L_{g}$ and $R_{g}$ the left and right translation
on the group $G$ by an element $g\in G$, respectively.} 
\[
Q\times G\rightarrow Q:\left(\left(x,h\right),g\right)\mapsto\left(x,R_{g}h\right)=\left(x,hg\right).
\]
They make $\pi:Q\rightarrow\mathcal{X}:\left(x,h\right)\mapsto x$
a right trivial principal fiber bundle with base $\mathcal{X}$ and
structure group $G$. For the lifted actions we shall use the notation
\[
TQ\times G\rightarrow TQ:\left(\left(x,h,\dot{x},\dot{h}\right),g\right)\mapsto\left(x,hg,\dot{x},\dot{h}g\right)
\]
and 
\[
T^{*}Q\times G\rightarrow T^{*}Q:\left(\left(x,h,y,\sigma\right),g\right)\mapsto\left(x,hg,y,\sigma g\right).
\]
Then, a principal connection $A:TQ\rightarrow\mathfrak{g}$ is given
by 
\begin{equation}
A\left(x,h,\dot{x},\dot{h}\right)=\Ad_{h^{-1}}(\mathcal{A}(x)\dot{x})+h^{-1}\dot{h}\label{aadh}
\end{equation}
where $\mathcal{A}$ is a $\mathfrak{g}$-valued $1$-form on $\mathcal{X}$,
i.e. $\mathcal{A}:\mathcal{X}\rightarrow T^{\ast}\mathcal{X}\otimes\mathfrak{g}$,
and it is given by the formula $\mathcal{A}\left(x\right)\dot{x}=A\left(x,e,\dot{x},0\right)$.
We shall say that the connection $A$ is the \textbf{trivial connection}
if $\mathcal{A}\left(x\right)=0$ for all $x\in\mathcal{X}$.

\bigskip{}
 Let us enumerate some identifications that we can do for these bundles. 
\begin{itemize}
\item The adjoint bundle $\widetilde{\mathfrak{g}}$ can be identified with
$\mathcal{X}\times\mathfrak{g}$ using the map 
\[
\widetilde{\mathfrak{g}}\rightarrow\mathcal{X}\times\mathfrak{g}:\left[\left(x,h\right),\xi\right]\longmapsto\left(x,\Ad_{h^{-1}}\xi\right),
\]
with inverse 
\[
\left(x,\xi\right)\longmapsto\left[\left(x,e\right),\xi\right].
\]
Analogously, $\widetilde{\mathfrak{g}}^{\ast}$ and $\mathcal{X}\times\mathfrak{g}^{\ast}$
can be identified by 
\[
\widetilde{\mathfrak{g}}^{\ast}\rightarrow\mathcal{X}\times\mathfrak{g}^{\ast}:\left[\left(x,h\right),\mu\right]\longmapsto\left(x,\Ad_{h}^{\ast}\mu\right),
\]
with inverse 
\[
\left(x,\mu\right)\longmapsto\left[\left(x,e\right),\mu\right].
\]
\item Using the above identifications, $T\mathcal{X}\oplus\widetilde{\mathfrak{g}}$
and $T^{\ast}\mathcal{X}\oplus\widetilde{\mathfrak{g}}^{\ast}$ are
naturally identified with $T\mathcal{X}\times\mathfrak{g}$ and $T^{\ast}\mathcal{X}\times\mathfrak{g}^{\ast}$,
respectively. As a consequence, $\alpha_{A}$ can be seen as the map
$\alpha_{A}:TQ/G\rightarrow T\mathcal{X}\times\mathfrak{g}$, given
by {[}recall \eqref{pma} and \eqref{aadh}{]} 
\[
\alpha_{A}\circ p\left(x,h,\dot{x},\dot{h}\right)=\left(\left(x,\dot{x}\right),\mathcal{A}\left(x\right)\dot{x}+h^{-1}\dot{h}\right).
\]
Then, the isomorphism $\hat{\alpha}_{A}:T^{*}Q/G\rightarrow T^{*}\mathcal{X}\times\mathfrak{g}^{*}$
is given by the formula (see Remark \ref{momentum}) 
\[
\hat{\alpha}_{A}\circ\hat{p}(x,h,y,\sigma)=\left(\left(x,y-\left(\mathcal{A}\left(x\right)\right)^{*}\left(\Ad_{h}^{*}(\sigma)\right)\right),\Ad_{h}^{*}{\bf J}\left(x,h,y,\sigma\right)\right).
\]
\item The curvature $B$ of the principal connection $A$ can be written
as 
\[
B\left(\left(x,e,\dot{x},0\right),\left(x,e,\delta x,0\right)\right)=d\mathcal{A}\left(\left(x,\dot{x}\right),\left(x,\delta x\right)\right)-\left[\mathcal{A}\left(x\right)\dot{x},\mathcal{A}\left(x\right)\delta x\right].
\]
Thus, from the very definition of $\widetilde{B}$ {[}see \eqref{bt}{]},
and identifying $\widetilde{\mathfrak{g}}$ and $\mathcal{X}\times\mathfrak{g}$,
we have that 
\begin{align}
\widetilde{B}\left(\left(x,\dot{x}\right),\left(x,\delta x\right)\right) & =\left(x,B\left(\left(x,e,\dot{x},0\right),\left(x,e,\delta x,0\right)\right)\right)\nonumber \\
 & =\left(x,d\mathcal{A}\left(\left(x,\dot{x}\right),\left(x,\delta x\right)\right)-\left[\mathcal{A}\left(x\right)\dot{x},\mathcal{A}\left(x\right)\delta x\right]\right).\label{btil}
\end{align}
\item The reduced hamiltonian $h$ can be seen as a map $h:T^{\ast}\mathcal{X}\times\mathfrak{g}^{\ast}\rightarrow\mathbb{R}$,
and its fiber and base derivatives as maps 
\[
\frac{\partial h}{\partial y}:T^{\ast}\mathcal{X}\times\mathfrak{g}^{\ast}\rightarrow T\mathcal{X}\ ,\quad\frac{\partial h}{\partial\bar{\mu}}:T^{\ast}\mathcal{X}\times\mathfrak{g}^{\ast}\rightarrow\mathcal{X}\times\mathfrak{g}
\]
and 
\[
\frac{\partial h}{\partial x}:T^{\ast}\mathcal{X}\times\mathfrak{g}^{\ast}\rightarrow T^{\ast}\mathcal{X},
\]
respectively. Following these observations, given $\left(x_{0},y_{0},\mu_{0}\right)\in T^{\ast}\mathcal{X}\times\mathfrak{g}^{\ast}$,
the element $\partial h/\partial\bar{\mu}\left(x_{0},y_{0},\mu_{0}\right)\in\mathcal{X}\times\mathfrak{g}$
is, essentially, the partial derivative of $h$ w.r.t. the vector
space variable $\mu\in\mathfrak{g}^{\ast}$. More precisely, we can
write 
\[
\frac{\partial h}{\partial\bar{\mu}}\left(x_{0},y_{0},\mu_{0}\right)=\left(x_{0},\frac{\partial h}{\partial\mu}\left(x_{0},y_{0},\mu_{0}\right)\right),
\]
where the second component is the usual derivative of $h(x_{0},y_{0},\mu)$
at $\mu_{0}$ as a function between the vector spaces $\mathfrak{g}^{*}$
and $\mathfrak{g}$. In these terms, given a curve $\varsigma\left(t\right)=\left(x\left(t\right),y\left(t\right),\mu\left(t\right)\right)$,
denoting by $\bar{\mu}\left(t\right)\in\widetilde{\mathfrak{g}}^{\ast}$
the curve $\left(x\left(t\right),\mu\left(t\right)\right)\in\mathcal{X}\times\mathfrak{g}^{\ast}$,
we have that 
\begin{equation}
\ad_{\frac{\partial h}{\partial\bar{\mu}}\left(\varsigma\left(t\right)\right)}^{\ast}\bar{\mu}\left(t\right)=\left(x(t),\ad_{\frac{\partial h}{\partial\mu}\left(\varsigma\left(t\right)\right)}^{\ast}\mu\left(t\right)\right),\label{addh}
\end{equation}
where the second $\ad^{*}$ is the usual coadjoint action of $\mathfrak{g}$
on $\mathfrak{g}^{*}$. Under the same identifications, the base derivative
of $h$ can be seen as a map 
\[
\frac{\partial h}{\partial x}:T^{*}\mathcal{X}\times\mathfrak{g}^{*}\rightarrow T^{*}\mathcal{X}.
\]
Recall that the latter is defined by an affine connection $\nabla$
on $T\mathcal{X}\oplus\widetilde{\mathfrak{g}}$ given as a sum $\nabla=\nabla^{\mathcal{X}}\oplus\nabla^{A}$. 
\item We can see the map $\varphi:T\mathcal{X}\rightarrow\widetilde{\mathfrak{g}}$,
defined in Section \ref{tc}, as a function $\varphi:T\mathcal{X}\rightarrow\mathcal{X}\times\mathfrak{g}$.
Suppose that $\varphi$ is related to the trivial connection $A$
and to the nonholonomic connection $A^{\bullet}$, and denote by $\mathcal{A}^{\bullet}$
to the $1$-form related to $A^{\bullet}$. In Ref. \cite{gz}, it
was proved that 
\begin{equation}
\varphi\left(x,v\right)=\left(x,-\mathcal{A}^{\bullet}\left(x\right)v\right).\label{fia}
\end{equation}
\item In the same Reference, the covariant derivative of a curve on $\widetilde{\mathfrak{g}}$
corresponding to $\nabla^{A}$ was calculated. From that, we can easily
write the covariant derivative of a curve $\bar{\mu}(t)\in\widetilde{\mathfrak{g}}^{*}$
(w.r.t. the affine connection dual to $\nabla^{A}$), as 
\begin{equation}
\frac{D\bar{\mu}}{Dt}\left(t\right)=\frac{D\left(x\left(t\right),\mu\left(t\right)\right)}{Dt}=\left(x\left(t\right),\mu'\left(t\right)+\ad_{\mathcal{A}\left(x\left(t\right)\right)x'\left(t\right)}^{\ast}\,\mu\left(t\right)\right).\label{cdgd}
\end{equation}
\item In addition, 
\begin{equation}
\left\langle \frac{\partial h}{\partial x},\delta x\right\rangle =\left\langle \frac{\partial^{c}h}{\partial x},\delta x\right\rangle +\left\langle \frac{\partial h}{\partial\mu},\ad_{\mathcal{A}\left(x\right)\delta x}^{*}\mu\right\rangle ,\label{bdc2}
\end{equation}
where $\partial^{c}h/\partial x$ is the base derivative of $h$ with
fixed $\mu$ and with respect to $\nabla^{\mathcal{X}}$. 
\end{itemize}
In the following we shall write the reduced equations (using one and
two connections) in the case of a general trivial bundle and then
we will consider some useful particular situations.

\paragraph{Case 1: General case.}

Suppose now that $A\neq A^{\bullet}$. Based on above observations,
the horizontal reduced equations for a trivial principal bundle are
{[}recall Eqs. \eqref{addh}, \eqref{cdgd} and \eqref{bdc2}{]}\footnote{For simplicity, we are omitting the dependency on $t$ and $\varsigma(t)$.}
\[
\left\langle \frac{Dy}{Dt}+\frac{\partial^{c}h}{\partial x}+\varphi^{*}\left(\mu'+\ad_{\mathcal{A}\left(x\right)x'}^{\ast}\mu-\ad_{\frac{\partial h}{\partial\mu}}^{\ast}\mu\right),\delta x^{\bullet}\right\rangle +\left\langle \mu,\widetilde{B}\left(\frac{\partial h}{\partial y},\delta x^{\bullet}\right)+\ad_{\mathcal{A}\left(x\right)\delta x^{\bullet}}\frac{\partial h}{\partial\mu}\right\rangle =0
\]
while the vertical reduced equations read 
\[
\left\langle \mu'+\ad_{\mathcal{A}\left(x\right)x'}^{\ast}\mu-\ad_{\frac{\partial h}{\partial\mu}}^{\ast}\mu,\eta^{\bullet}\right\rangle =0,
\]
where $\eta^{\bullet}$ is seen as a curve on $\mathfrak{g}$, and
$\varphi$ and $\widetilde{B}$ as maps taking values in $\mathfrak{g}$,
rather than $\widetilde{\mathfrak{g}}$ {[}recall Eqs. \eqref{fia}
and \eqref{btil}{]}. If we choose to work with only one connection,
i.e. we take $A=A^{\bullet}$ (and accordingly $\varphi=0$), then
the equations take the form 
\[
\left\langle \frac{Dy}{Dt}+\frac{\partial^{c}h}{\partial x},\delta x\right\rangle +\left\langle \mu,\widetilde{B}\left(\frac{\partial h}{\partial y},\delta x\right)+\ad_{\mathcal{A}\left(x\right)\delta x}\frac{\partial h}{\partial\mu}\right\rangle =0
\]
and 
\[
\left\langle \mu'+\ad_{\mathcal{A}\left(x\right)x'}^{\ast}\mu-\ad_{\frac{\partial h}{\partial\mu}}^{\ast}\mu,\eta\right\rangle =0.
\]

\begin{rem}
\label{notat} In the above expression we are omitting the dot $\bullet$,
since we have only one connection and we do not need to make any distinction
(as in Section \ref{reducedextremal}). 
\end{rem}

\paragraph{Case 2: Choosing $A$ as the trivial connection.}

Assume now that we choose $A$ as the trivial connection on $\mathcal{X}\times G$.
Hence $\mathcal{A}=0$, which implies that $\widetilde{B}=0$, and
consequently the reduced equations read 
\begin{equation}
\left\langle \frac{Dy}{Dt}+\frac{\partial^{c}h}{\partial x}+\varphi^{*}\left(\mu'-\ad_{\frac{\partial h}{\partial\mu}}^{\ast}\mu\right),\delta x^{\bullet}\right\rangle =0\label{redhoreqC2}
\end{equation}
and 
\begin{equation}
\left\langle \mu'-\ad_{\frac{\partial h}{\partial\mu}}^{\ast}\mu,\eta^{\bullet}\right\rangle =0,\label{redvereqC2}
\end{equation}
If in addition we use Eq. \eqref{fia}, the horizontal equations can
be written 
\[
\left\langle \frac{Dy}{Dt}+\frac{\partial^{c}h}{\partial x},\delta x^{\bullet}\right\rangle -\left\langle \mu'-\ad_{\frac{\partial h}{\partial\mu}}^{\ast}\mu,\mathcal{A}^{\bullet}\left(x\right)\delta x^{\bullet}\right\rangle =0.
\]
We emphasize that this last simplification cannot be done in the one-connection-approach,
because $A^{\bullet}$ does not necessarily coincide with the trivial
connection.

\paragraph{Case 3: $T^{*}\mathcal{X}$ is a trivial bundle and $A$ is again
the trivial connection.}

If $T^{*}\mathcal{X}$ is trivial, then $\partial h/\partial y$ can
be seen as a partial derivative in a linear space. In addition, if
we choose $\nabla^{\mathcal{X}}$ as the trivial affine connection,
then the covariant derivative is a standard derivative of a vector
variable with respect to $t$, i.e. 
\[
\frac{Dy}{Dt}=y'.
\]
On the other hand, $\partial^{c}h/\partial x$ is also a standard
partial derivative: $\partial h/\partial x$. Therefore, the reduced
equations in the two-connection-approach get simplified as 
\[
\left\langle y'+\frac{\partial h}{\partial x},\delta x^{\bullet}\right\rangle -\left\langle \mu'-\ad_{\frac{\partial h}{\partial\mu}}^{\ast}\mu,\mathcal{A}^{\bullet}\left(x\right)\delta x^{\bullet}\right\rangle =0
\]
and 
\[
\left\langle \mu'-\ad_{\frac{\partial h}{\partial\mu}}^{\ast}\mu,\eta^{\bullet}\right\rangle =0.
\]

\paragraph{The case of HOCS. }

Similar calculations can be made for HOCSs. We just must replace the
standard connections by cotangent $l$-connections $A:T^{\left(l\right)}T^{*}Q\times_{Q}TQ\rightarrow\mathfrak{g}$,
which can be written 
\begin{align*}
A\left(\zeta;x,h,\dot{x},\dot{h}\right) & =\Ad_{h}\left(\mathcal{A}\left(\left[\zeta\right]\right)\,\dot{x}\right)+\dot{h}\,h^{-1},
\end{align*}
with $\mathcal{A}:\left.T^{\left(l\right)}T^{*}Q\right/G\rightarrow T^{\ast}\mathcal{X}\otimes\mathfrak{g}$
given by 
\[
\mathcal{A}\left(\left[\zeta\right]\right)\,\dot{x}=A\left(\zeta;x,e,\dot{x},0\right).
\]
Also, the map $\varphi$ must be replaced by the maps $\varphi^{\left[\zeta\right]}$.
In the case in which one of the connection is trivial and the other
is the higher-order $l$-connection $A^{\bullet}$, the maps $\varphi^{\left[\zeta\right]}$
are given by (under usual identifications) 
\[
\varphi^{\left[\zeta\right]}\left(x,\dot{x}\right)=\left(x,-\mathcal{A}^{\bullet}\left[\zeta\right]\,\dot{x}\right).
\]


\section{A ball rolling without sliding over another ball}

Let us consider now a mechanical system consisting of two balls $B_{1}$
and $B_{2}$ of radii $r_{1}$ and $r_{2}$, respectively, in the
presence of gravity (see Figure ...). Suppose that $r_{1}>r_{2}$
and that $B_{2}$ is rolling without sliding over the surface of $B_{1}$.
Assume also that the center of $B_{1}$ is fixed with respect to a
given inertial reference system, and that $B_{1}$ can freely rotate
around its center. Now, consider the following control problem: stabilize
asymptotically the smaller ball $B_{2}$ on the top of the bigger
one $B_{1}$ by making a torque on $B_{1}$. (Such a torque would
be the feedback controller). This problem can be addressed by imposing
a so-called Lyapunov constraint \cite{gym}. In such a case, the mentioned
torque is given by the related constraint force. Since we have more
constraint forces directions than constraints, the system of equations
will be underdeterminated. Here, what it is important for us is that
the original Hamiltonian system, with the nonholonomic (rolling-without-sliding)
constraint, the Lyapunov constraint and the torque direction, all
together, define a HOCS (see Ref. \cite{gmp}). Moreover, we shall
see immediately that such a HOCS can be chosen $SO(3)$-invariant
(in the sense of Section \ref{HHOCSsym}). The purpose of the present
section is to apply the reduction procedure developed previously to
this kind of HOCS. Concretely, our main aim is to find an expression
of the horizontal and vertical HdP equations for it.

\bigskip{}

To begin with, let us denote by $R$ and $I_{1}$ the rotation matrix
and the moment of inertia of $B_{1}$, respectively, and let us indicate
by $C$, $I_{2}$ and $m_{2}$ the rotation matrix, the moment of
inertia and the mass of $B_{2}$. In addition, denote by $\mathbf{e}$
the unit vector with origin in the center of $B_{1}$ and pointing
in the direction of the center of $B_{2}$.
\begin{itemize}
\item \textbf{Configuration space.}

It is clear that the configuration of the system can be described
by the triple $(R,\mathbf{e},C)$, i.e. its configuration space is
given by 
\[
Q=SO(3)\times S^{2}\times SO(3).
\]
The elements of the tangent bundle $TQ$ will be denoted $(R,\dot R,\mathbf{e},\dot{\mathbf{e}},C,\dot C)$, except when we refer to elements in $C_V$, where the notation $(R,\delta R,\mathbf{e},\delta \mathbf{e},C,\delta C)$ will be used instead.
To make computations easier, we will use the following identifications:
\begin{equation}
TSO(3)\simeq SO(3)\times\mathfrak{so}(3)\simeq SO(3)\times\mathbb{R}^{3}.\label{idtan}
\end{equation}
The first one is the well-known trivialization by right translations.
The second identification uses a Lie algebra isomorphism 
\[
\widehat{}:(\mathbb{R}^{3},\times)\rightarrow(\mathfrak{so}(3),[\cdot,\cdot])
\]
given by\footnote{We identify $\mathfrak{so}(3)$ with the set of skew-symmetric matrices of dimension 3.}
\[
\eta=(\eta^{1},\eta^{2},\eta^{3})\mapsto\widehat{\eta}=\begin{pmatrix}0 & -\eta^{3} & \eta^{2}\\
\eta^{3} & 0 & -\eta^{1}\\
-\eta^{2} & \eta^{1} & 0
\end{pmatrix}.
\]
In $\mathbb{R}^{3}$ the lie bracket is given by the cross product
$\times$ of vectors. Under these identifications, we have 
\[
(C,\dot C)\simeq(C,\dot CC^{-1})\simeq(C,\xi)\quad\text{and}\quad(R,\dot R)\simeq(R,\dot RR^{-1})\simeq(R,\eta),
\]
where $\widehat{\xi}=\dot CC^{-1}$ and $\widehat{\eta}=\dot RR^{-1}$.
As a consequence, 
\[
TQ\simeq SO(3)\times\mathbb{R}^{3}\times TS^{2}\times SO(3)\times\mathbb{R}^{3}.
\]
Since $S^{2}$ is a submanifold of the euclidean space $\mathbb{R}^{3}$,
we will sometimes see the space $T_{\mathbf{e}}S^{2}$ as a linear
subspace of $\mathbb{R}^{3}$. Analogously, the cotangent space will
be identified as 
\begin{equation}
T^{*}Q\simeq SO(3)\times\mathbb{R}^{3}\times T^{*}S^{2}\times SO(3)\times\mathbb{R}^{3}.\label{idcot}
\end{equation}
A covector at $(R,\mathbf{e},C)$ will be written $(R,\pi,\mathbf{e},\sigma,C,\gamma)\in T^{*}Q$,
with $\pi,\gamma\in\mathbb{R}^{3}$.

\bigskip{}

\item \textbf{Hamiltonian function.}

Now, let us describe the dynamics of the system. The Lagrangian $L:TQ\rightarrow\mathbb{R}$
is given by 
\[
L(R,\eta,\mathbf{e},\dot{\mathbf{e}},C,\xi)=\frac{1}{2}I_{1}\eta^{2}+\frac{1}{2}m_{2}\,\dot{\mathbf{e}}^{2}+\frac{1}{2}I_{2}\xi^{2}-m_{2}g\,\mathbf{e}\cdot\mathbf{z},
\]
where $\cdot$ and $(\cdot)^{2}$ denote the euclidean inner product
and the squared euclidean norm on $\mathbb{R}^{3}$, respectively;
$g$ is the acceleration of gravity and $\mathbf{z}=(0,0,1)$ is the
vertical unit vector pointing upwards (see Figure ...). In order to
obtain the Hamiltonian $H:T^{*}Q\rightarrow\mathbb{R}$ of the system,
we must use the Legendre transform $\mathbb{F}L:TQ\rightarrow T^{*}Q$,
given by 
\[
\mathbb{F}L(R,\eta,\mathbf{e},\dot{\mathbf{e}},C,\xi)=(R,\pi,\mathbf{e},\sigma,C,\gamma)=(R,I_{1}\eta,\mathbf{e},m_{2}\dot{\mathbf{e}},C,I_{2}\xi).
\]
It is easy to show that 
\begin{equation}
H(R,\pi,\mathbf{e},\sigma,C,\gamma)=\frac{1}{2I_{1}}\pi^{2}+\frac{1}{2m_{2}}\sigma^{2}+\frac{1}{2I_{2}}\gamma^{2}+m_{2}g\,\mathbf{e}\cdot\mathbf{z}.\label{Hs}
\end{equation}

\bigskip{}

\item \textbf{Lyapunov constraint and related constraint force.}

Given two non-negative functions $V,\mu\in C^{\infty}(T^{*}Q)$, consider
the submanifold 
\begin{equation}
\mathcal{P}^{\text{Lyap}}\coloneqq\bigcup_{\alpha\in T^{*}Q}\left\{ w\in T_{\alpha}T^{*}Q\;:\;\left<dV(\alpha),w\right>=-\mu(\alpha)\right\} \subset TT^{*}Q.\label{plyap}
\end{equation}
If $V$ is positive-definite around some point of $T^{*}Q$, according
to Ref. \cite{gym}, above submanifold defines a Lyapunov constraint.
Observe that the latter is a second order constraint, i.e. $k=2$
(see Def. \ref{def:hocs}). We shall assume that $V$ is of the form
\begin{equation}
V(R,\pi,\mathbf{e},\sigma,C,\gamma)=\frac{1}{2}(\pi^{t},\sigma^{t},\gamma^{t})\Phi(R,\mathbf{e})(\pi,\sigma,\gamma)+v(R,\mathbf{e}),\label{vlyap}
\end{equation}
where $\Phi$ is a positive-definite matrix depending smoothly on
$(R,\mathbf{e})$ and $v\in C^{\infty}(SO(3)\times S^{2})$ is nonnegative.

If we want to implement this constraint by making a torque on the
ball $B_{1}$, then the space of constraint forces and its related
variational constraints would be, respectively, 
\[
F_{V}^{\text{Lyap}}(R,\mathbf{e},C):=\mathbb{R}^{3}\times\{0\}\times\{0\}\subseteq T_{R}^{*}SO\left(3\right)\times T_{\mathbf{e}}^{*}S^{2}\times T_{C}^{*}SO\left(3\right)
\]
and {[}see Eq. \eqref{fcv}{]} 
\[
C_{V}^{\text{Lyap}}(R,\mathbf{e},C)=\left(F_{V}^{\text{Lyap}}(R,\mathbf{e},C)\right)^{\circ}=\{0\}\times T_{\mathbf{e}}S^{2}\times\mathbb{R}^{3}.
\]

\bigskip{}

\item \textbf{Rolling constraint and d'Alembert's Principle.}

The rolling constraint in the Lagrangian formulation, and using the
notation introduced above, is given by the submanifold 
\[
C_{K}^{\textrm{Rol}}\coloneqq\left\{ (R,\eta,\mathbf{e},\dot{\mathbf{e}},C,\xi)\in TQ\;:\;\dot{\mathbf{e}}=\frac{1}{r_{1}+r_{2}}\left(r_{1}\eta+r_{2}\xi\right)\times\mathbf{e}\right\} .
\]
To obtain the Hamiltonian counterpart, it is enough to perform the
Legendre transform to find 
\begin{equation}
\mathcal{D}^{\textrm{Rol}}\coloneqq\mathbb{F}L(C_{K}^{\textrm{Rol}})=\left\{ (R,\pi,\mathbf{e},\sigma,C,\gamma)\in T^{*}Q\;:\;\frac{1}{m_{2}}\sigma=\frac{1}{r_{1}+r_{2}}\left(\frac{r_{1}}{I_{1}}\pi+\frac{r_{2}}{I_{2}}\gamma\right)\times\mathbf{e}\right\} .\label{drod}
\end{equation}
Finally, assuming d'Alembert's Principle, i.e. assuming that the space
of constraint forces implementing above constraint is given by 
\[
F^{\textrm{Rol}}\coloneqq\left(C_{K}^{\textrm{Rol}}\right)^{\circ},
\]
then the set of related variational constraints will be given by $C_{V}^{\textrm{Rol}}:=C_{K}^{\textrm{Rol}}$.

\bigskip{}

\item \textbf{Resulting HOCS.}

The set defined by all the kinematic constraints is 
\[
\mathcal{P}\coloneqq\tau_{T^{*}Q}^{-1}\left(\mathcal{D}^{\textrm{Rol}}\right)\cap\mathcal{P}^{\text{Lyap}},
\]
while the set of variational ones reads 
\[
C_{V}\coloneqq C_{V}^{\textrm{Rol}}\cap C_{V}^{\text{Lyap}}=\left\{ (R,0,\mathbf{e},\delta\mathbf{e},C,\xi)\in TQ\;:\;\delta\mathbf{e}=r_{12}\left(\xi\times\mathbf{e}\right)\right\} ,
\]
where $r_{12}=\frac{r_{2}}{r_{1}+r_{2}}$. This data, together with
the Hamiltonian function \eqref{Hs}, gives rise to the HOCS $\left(H,\mathcal{P},C_{V}\right)$
(see Remarks \ref{gh} and \ref{vcv}).

\bigskip{}

\item \textbf{Symmetry group.}

We are now ready to define the symmetry of the system. Let us consider
the right action of $SO(3)$ on $Q$, $\rho:Q\times SO(3)\rightarrow Q$,
given by 
\[
\rho\left((R,\mathbf{e},C),B\right)=(R,\mathbf{e},CB).
\]
It is essentially the right translation of $SO(3)$ onto itself. Thus,
it is a free action making the map $SO(3)\times S^{2}\times SO(3)\rightarrow SO(3)\times S^{2}=\mathcal{X}$
into a trivial principal bundle with $SO(3)$ as a structure group.
It is easy to prove, using the identifications \eqref{idtan} and
\eqref{idcot}, that the lifted actions to $TQ$ and $T^{*}Q$ are
given by 
\[
\rho_{B*}(R,\eta,\mathbf{e},\dot{\mathbf{e}},C,\xi)=(R,\eta,\mathbf{e},\dot{\mathbf{e}},CB,\xi)\quad\text{and}\quad\hat{\rho}_{B}(R,\pi,\mathbf{e},\sigma,C,\gamma)=(R,\pi,\mathbf{e},\sigma,CB,\gamma).
\]
Both the Lagrangian and the Hamiltonian functions are clearly invariant
with respect to these actions. On the other hand, a simple calculation
shows that $\mathcal{P}$ and $C_{V}$ are invariant too. Therefore,
$(H,\mathcal{P},C_{V})$ is indeed a $SO(3)$-invariant HOCS.

\bigskip{}

\item \textbf{Generalized nonholonomic connection and associated maps.}

Since the variational constraints $C_{V}$ of our system are given
by a subbundle of $TQ$, i.e. we have variations of order $l=1$,
to decompose them into vertical and horizontal parts, it is enough
to consider a standard connection. We will now construct the generalized nonholonomic connection associated to our problem, following the steps described at the beginning of Section \ref{hocc}. To do that, we need
to construct the spaces $\mathcal{S}$, $\mathcal{U}$, $\mathcal{T}$
and $\mathcal{R}$. Observe that the vertical space is given by 
\[
\mathbb{V}=\{0\}\times\{0\}\times\mathbb{R}^{3}=\{(R,0,\mathbf{e},0,C,\xi)\in TQ\;:\;\xi\in\mathbb{R}^{3}\}
\]
and 
\[
\mathcal{S}=\mathbb{V}\cap C_{V}=\{(R,0,\mathbf{e},0,C,\xi)\in TQ\;:\;\xi\in span\{\mathbf{e}\}\}.
\]
In order to calculate the spaces $\mathcal{U}$ and $\mathcal{T}$,
we must use the Riemannian metric of the kinetic term of $H$ to take
the orthogonal complements. We obtain\footnote{The superscript $\perp$ denote orthogonal complement with respect
to the euclidean inner product in $\mathbb{R}^{3}$.} 
\[
\mathcal{U}=\{(R,0,\mathbf{e},0,C,\xi)\in TQ\;:\;\xi\in span\{\mathbf{e}\}^{\perp}\}
\]
and 
\[
\mathcal{T}=\{(R,0,\mathbf{e},\delta\mathbf{e},C,\xi)\in TQ\;:\;\delta\mathbf{e}=r_{12}\left(\xi\times\mathbf{e}\right),\;\xi\in span\{\mathbf{e}\}^{\perp}\}.
\]
An easy calculation shows that $\mathcal{T}$ may be written as 
\[
\mathcal{T}=\left\{ \left(R,0,\mathbf{e},\delta\mathbf{e},C,\frac{1}{r_{12}}(\mathbf{e}\times\delta\mathbf{e})\right)\in TQ\;:\;(\mathbf{e},\delta\mathbf{e})\in TS^{2}\right\} .
\]
Finally, since $\mathcal{R}$ is the orthogonal complement of $C_{V}+\mathbb{V}$
in $TQ$ and 
\[
C_{V}+\mathbb{V}=\mathbb{V}\oplus\mathcal{T}=\left\{ (R,0,\mathbf{e},\delta\mathbf{e},C,\xi)\in TQ\;:\;(\mathbf{e},\delta\mathbf{e})\in TS^{2},\;\xi\in\mathbb{R}^{3}\right\} =\{0\}\oplus TS^{2}\oplus\mathbb{R}^{3},
\]
we have that 
\[
\mathcal{R}=\mathbb{R}^{3}\oplus\{0\}\oplus\{0\}.
\]
Gathering all the previous expressions, we define the wanted connection
as that given by the horizontal space 
\[
\mathbb{H}^{\bullet}=\mathcal{T}\oplus\mathcal{R}=\left\{ \left(R,\eta,\mathbf{e},\delta\mathbf{e},C,\frac{1}{r_{12}}(\mathbf{e}\times\delta\mathbf{e})\right)\in TQ\;:\;(\mathbf{e},\delta\mathbf{e})\in TS^{2},\;\eta\in\mathbb{R}^{3}\right\} .
\]
Consequently, the connection form is given by 
\[
A^{\bullet}(R,\eta,\mathbf{e},\delta\mathbf{e},C,\xi)=C^{-1}\left(\xi-\frac{1}{r_{12}}(\mathbf{e}\times\delta\mathbf{e})\right).
\]

\bigskip{}

On the other hand, based on the calculations of the previous section,
the trivial connection may be written 
\[
A(R,\eta,\mathbf{e},\delta\mathbf{e},C,\xi)=C^{-1}\xi.
\]

\bigskip{}

With these expressions at hand, we can write the Atiyah isomorphisms
associated with $A^{\bullet}$ and $A$, respectively, as follows
\[
\alpha_{A}\left([R,\eta,\mathbf{e},\delta\mathbf{e},C,\xi]_{SO(3)}\right)=(R,\eta,\mathbf{e},\delta\mathbf{e},\xi)
\]
and 
\[
\alpha_{A^{\bullet}}([R,\eta,\mathbf{e},\delta\mathbf{e},C,\xi]_{SO(3)})=\left(R,\eta,\mathbf{e},\delta\mathbf{e},\xi-\frac{1}{r_{12}}(\mathbf{e}\times\delta\mathbf{e})\right).
\]
Thus, the map $\varphi:T\mathcal{X}\rightarrow\widetilde{\mathfrak{g}}=\mathcal{X}\times\mathbb{R}^{3}$
is given by the formula 
\[
\varphi(R,\eta,\mathbf{e},\delta\mathbf{e})=\left(R,\mathbf{e},\frac{1}{r_{12}}(\mathbf{e}\times\delta\mathbf{e})\right).
\]
(Recall that $\mathcal{X}=SO(3)\times S^{2}$). As we pointed out
in Remark \ref{momentum}, the isomorphisms $\hat{\alpha}_{A}$ and
$\hat{\alpha}_{A^{\bullet}}$ can be written in terms of the momentum
map 
\[
J(R,\pi,\mathbf{e},\sigma,C,\gamma)=C^{-1}\gamma
\]
as 
\[
\hat{\alpha}_{A}([R,\pi,\mathbf{e},\sigma,C,\gamma]_{SO(3)})=(\pi,\sigma)\oplus[q,C^{-1}\gamma]_{SO(3)}=(R,\pi,\mathbf{e},\sigma,\gamma)
\]
and 
\[
\hat{\alpha}_{A^{\bullet}}([R,\pi,\mathbf{e},\sigma,C,\gamma]_{SO(3)})=\left(\pi,\sigma+\frac{1}{r_{12}}\gamma\times\mathbf{e}\right)\oplus[q,C^{-1}\gamma]_{SO(3)}=\left(R,\pi,\mathbf{e},\sigma+\frac{1}{r_{12}}\gamma\times\mathbf{e},\gamma\right).
\]

\bigskip{}

\item \textbf{Reduced data.}

Using the Atiyah isomorphism $\alpha_{A^{\bullet}}$, we can write
explicit expressions for the horizontal and vertical variational constraints
$\mathfrak{C}_{V}^{\text{hor}}\subset T\mathcal{X}$ and $\mathfrak{C}_{V}^{\text{ver}}\subset\mathcal{X}\times\mathbb{R}^{3}$
as follows: 
\begin{equation}
\mathfrak{C}_{V}^{\text{hor}}=\alpha_{A^{\bullet}}\circ p(\mathcal{T})=\left\{ \left(R,0,\mathbf{e},\delta\mathbf{e}\right)\in T\mathcal{X}\;:\;(\mathbf{e},\delta\mathbf{e})\in TS^{2}\right\} \label{cvhorex}
\end{equation}
and 
\begin{equation}
\mathfrak{C}_{V}^{\text{ver}}=\alpha_{A^{\bullet}}\circ p(\mathcal{S})=\left\{ (R,\mathbf{e},\xi)\in\mathcal{X}\times\mathbb{R}^{3}\;:\;\xi\in span\{\mathbf{e}\}\right\} .\label{cvverex}
\end{equation}

In order to write down the reduced Hamiltonian $h:T^{*}\mathcal{X}\times\mathbb{R}^{3}\rightarrow\mathbb{R}$
and the reduced kinematic constraints $\mathfrak{P}$, we shall use
the trivial connection. This gives rise exactly to the expressions
that we already have, i.e. $h$ is given by \eqref{Hs} and the reduced
kinematic constraints $\mathfrak{P}$ by the Equations \eqref{plyap},
\eqref{vlyap} and \eqref{drod}.

\bigskip{}

\item \textbf{Higher Order HdP Equations.}

It only remains to write the reduced equations. To do that, we will
use the computations of the previous section. Since $\mathcal{X}=SO(3)\times S^{2}$,
we have $T^{*}\mathcal{X}=(SO(3)\times\mathbb{R}^{3})\times_{\mathcal{}}T^{*}S^{2}$
and we can consider on $\mathcal{X}$ a connection of the form $\nabla^{\mathcal{X}}=\nabla^{SO(3)}\times\nabla^{S^{2}}$,
being $\nabla^{SO(3)}$ the trivial affine connection. Under this
assumption, we can write 
\[
\frac{Dy}{Dt}(t)=\frac{D(\pi,\sigma)}{Dt}(t)=\left(\pi'(t),\frac{D\sigma}{Dt}(t)\right),
\]
where $\pi'(t)$ is the usual derivative of a curve in $\mathbb{R}^{3}$
and $D\sigma/Dt$ is the covariant derivative related to $\nabla^{S^{2}}$,
and write
\[
\frac{\partial^{c}h}{\partial x}=\left(\frac{\partial^{c}h}{\partial R},\frac{\partial^{c}h}{\partial\mathbf{e}}\right)=\left(0,\frac{\partial^{c}h}{\partial\mathbf{e}}\right),
\]
where ${\displaystyle {\frac{\partial^{c}h}{\partial R}}}$ and ${\displaystyle \frac{\partial^{c}h}{\partial\mathbf{e}}}$
are base derivative of $h$ respect to $\nabla^{SO(3)}$ and $\nabla^{S^{2}}$,
respectively. The first component is zero because $\nabla^{SO(3)}$
is trivial and $h$ is independent on $R$. If we take on $S^{2}$
the Levi-Civita connection induced by the Euclidean metric on $\mathbb{R}^{3}$,
and identify each space $T_{\mathbf{e}}^{*}S^{2}$ with a linear subspace
of $\mathbb{R}^{3}$, it is easy to see that
\begin{equation}
\frac{D\sigma}{Dt}(t)=\sigma'(t)-\left(\mathbf{e}(t)\cdot\sigma'(t)\right)\,\mathbf{e}(t)\label{covsig}
\end{equation}
and
\begin{equation}
\frac{\partial^{c}h}{\partial\mathbf{e}}=m_{2}g\,\left(\mathbf{z}-\left(\mathbf{e}\cdot\mathbf{z}\right)\,\mathbf{e}\right).\label{covh}
\end{equation}
On the other hand, under the same identification, the fiber derivatives
of $h$ are in fact usual derivatives, i.e. 
\begin{equation}
\frac{\partial h}{\partial\gamma}=\nabla_{\gamma}h=\frac{\gamma}{I_{2}}\in\mathbb{R}^{3}\label{deg}
\end{equation}
and
\[
\frac{\partial h}{\partial y}=\left(\nabla_{\pi}h,\nabla_{\sigma}h\right)=\left(\frac{\pi}{I_{1}},\frac{\sigma}{m_{2}}\right)\in\mathbb{R}^{3}\times\mathbb{R}^{3},
\]
where the subscripts denote the variable respect to which we are differentiating.
Finally, notice that 
\begin{equation}
\ad_{\xi}^{*}\gamma=\gamma\times\xi\label{fiad1}
\end{equation}
and
\begin{equation}
\varphi^{*}(R,\mathbf{e},\gamma)=\left(R,\mathbf{e},0,\frac{1}{r_{12}}\gamma\times\mathbf{e}\right).\label{fiad}
\end{equation}

\bigskip{}

Gathering all these elements, we are ready now to write down the reduced
equations. Taking into account that $\delta R^{\bullet}=0$ and $\delta\mathbf{e}^{\bullet}$
is arbitrary {[}recall Eq. \eqref{cvhorex}{]}, the horizontal equations
read {[}see \eqref{redhoreqC2}, \eqref{deg}, \eqref{fiad1} and
\eqref{fiad}{]} 
\[
\frac{D\sigma}{Dt}+\frac{\partial^{c}h}{\partial\mathbf{e}}+\frac{1}{r_{12}}\gamma'\times\mathbf{e}=0,
\]
which implies {[}see \eqref{covsig} and \eqref{covh}{]}
\[
\sigma'(t)-\left(\mathbf{e}(t)\cdot\sigma'(t)\right)\,\mathbf{e}(t)+m_{2}g\,\left(\mathbf{z}-\left(\mathbf{e}(t)\cdot\mathbf{z}\right)\,\mathbf{e}(t)\right)+\frac{1}{r_{12}}\gamma'(t)\times\mathbf{e}(t)=0.
\]

On the other hand, according to the Eqs. \eqref{redvereqC2}, \eqref{cvverex},
\eqref{deg} and \eqref{fiad1}, the vertical equations are 
\[
\left\langle \gamma',\eta^{\bullet}\right\rangle =0,\qquad\forall\,\eta^{\bullet}\in span\{\mathbf{e}\},
\]
or equivalently,
\[
\gamma'(t)\cdot\mathbf{e}(t)=0.
\]

\end{itemize}

\end{document}